\documentclass[format=acmsmall, review=false, authorversion=true]{acmart}

\pdfoutput=1

\usepackage{booktabs} 

 \acmVolume{9}
 \acmNumber{4}
 \acmArticle{39}
 \acmYear{2017}
 \acmMonth{2}
 \copyrightyear{2017}
\acmArticleSeq{9}

\setcopyright{rightsretained}

\acmDOI{0000001.0000001}

\received{February 2017}
\received[revised]{?}
\received[accepted]{?}

\usepackage[utf8]{inputenc}
\usepackage[T1]{fontenc}

\usepackage{amsmath}
\usepackage{amssymb}
\usepackage{graphicx}
\usepackage{xcolor}
\usepackage{enumitem}

\usepackage{array}

\usepackage{stmaryrd}

\usepackage{url}

\usepackage{microtype}

\usepackage{tikz}

\usetikzlibrary{shapes}
\usetikzlibrary{positioning,calc}

\tikzstyle{buffer} = [rectangle split, rectangle split parts=2, draw, rectangle split horizontal,
    text width=5em, text centered, rounded corners, minimum height=3em]
\tikzstyle{chunk} = [circle, draw, text centered, text width=1em]
\tikzstyle{line} = [draw, -latex]

\newcommand*{\var}[1]{\mathord{\mathit{#1}}}

\newcommand*{\fun}[1]{\mathord{\mathit{#1}}}
\newcommand*{\constr}[1]{\mathord{\mathit{#1}}}

\newcommand{\consts}{\mathcal{C}}
\newcommand{\vars}{\mathcal{V}}
\newcommand{\test}[3]{\texttt{=}(\var{#1},\var{#2},\var{#3})}

\newcommand{\DDelta}{\mathbf{\Delta}}
\newcommand{\DDeltap}[1]{\mathbf{\Delta}_{#1}^\circledcirc}
\newcommand{\chunkmerge}{\circ}

\newcommand{\nil}{\mathtt{nil}}
\newcommand{\chunk}{\mathtt{chunk}}

\newcommand{\cvar}{\fun{cvar}}

\newcommand{\actionreq}[3]{\texttt{+}(\var{#1},\var{#2},\var{#3})}
\newcommand{\actionmod}[3]{\texttt{=}(\var{#1},\var{#2},\var{#3})}

\newcommand{\fvars}[1]{\fun{vars}(#1)}
\newcommand{\actions}{\mathcal{A}}
\newcommand{\states}{\mathcal{S}}

\newcommand{\utrans}{\rightarrowtail}

\newcommand{\applynorule}{\fun{update}}
\newcommand{\applyrule}{\fun{apply}}

\newcommand{\cogstates}{\Gamma}
\newcommand{\cogstatespart}{\Gamma_{\mathord{\mathrm{part}}}}
\newcommand{\rules}{\Sigma}
\newcommand{\traces}{\Pi}
\newcommand{\delayrule}{\delta}

\newcommand{\addinfo}{\upsilon}
\newcommand{\addinfos}{\Upsilon}

\newcommand{\types}{\mathbb{T}}

\newcommand{\buffers}{\mathbb{B}}

\newcommand{\val}{\fun{val}}
\newcommand{\id}{\fun{id}}

\newcommand{\matches}{\sqsubseteq}

\newcommand{\comb}{\sqcup}
\newcommand{\bigcomb}{\bigsqcup}

\newcommand{\newtime}{\vartheta}

\newcommand{\veryabstract}{\mathit{va}}
\newcommand{\abstr}{\mathit{abs}}

\begin{document}

\title[An Operational Semantics for ACT-R and its Translation to CHR]{An Operational Semantics for the Cognitive Architecture ACT-R and its Translation to Constraint Handling Rules}  
\author{Daniel Gall}
\author{Thom Frühwirth}
\affiliation{%
  \institution{Ulm University}
  \department{Institute of Software Engineering and Programming Languages}
  \city{Ulm}
  \postcode{89069}
  \country{Germany}
}

\begin{abstract}
Computational psychology has the aim to explain human cognition by computational models of cognitive processes. The cognitive architecture Adaptive Control of Thought -- Rational (ACT-R) is popular to develop such models. Although ACT-R has a well-defined psychological theory and has been used to explain many cognitive processes, there are two problems that make it hard to reason formally about its cognitive models: First, ACT-R lacks a computational formalization of its underlying production rule system and secondly, there are many different implementations and extensions of ACT-R with many technical artifacts complicating formal reasoning even more.

This paper describes a formal operational semantics -- the \emph{very abstract semantics} -- that abstracts from as many technical details as possible keeping it open to extensions and different implementations of the ACT-R theory. In a second step, this semantics is refined to define some of its abstract features that are found in many implementations of ACT-R -- called the \emph{abstract semantics}. It concentrates on the procedural core of ACT-R and is suitable for analysis of the general transition system since it still abstracts from details like timing, the sub-symbolic layer of ACT-R or conflict resolution.

Furthermore, a translation of ACT-R models to the declarative programming language Constraint Handling Rules (CHR) is defined. This makes the abstract semantics an executable specification of ACT-R. CHR has been used successfully to embed other rule-based formalisms like graph transformation systems or functional programming. There are many theoretical results and practical tools that support formal reasoning about and analysis of CHR programs. The translation of ACT-R models to CHR is proven sound and complete w.r.t. the abstract operational semantics of ACT-R. This paves the way to analysis of ACT-R models through CHR analysis results and tools. Therefore, to the best of our knowledge, our abstract semantics is the first abstract formulation of ACT-R suitable for both analysis and execution.
\end{abstract}

\ccsdesc[500]{Theory of computation~Operational semantics}
\ccsdesc[500]{Computing methodologies~Cog\-ni\-tive science}
\ccsdesc[500]{Applied computing~Psychology}
\ccsdesc[500]{Software and its engineering~Constraint and logic languages}
\ccsdesc[500]{Software and its engineering~Semantics}

%
%

\keywords{Computational cognitive modeling, ACT-R, operational semantics, source to source transformation, Constraint Handling Rules}


\maketitle


\section{Introduction}

Computational cognitive modeling tries to explore human cognition by building detailed computational models of cognitive processes \cite{sun_introduction_2008}. Cognitive architectures support the modeling process by providing a formal, well-investigated theory of cognition that allows for building cognitive models of specific tasks and cognitive features.

Currently, computational cognitive modeling architectures as well as the implementations of cognitive models are typically ad-hoc constructs. They lack a formalization from the computer science point of view. For instance, \emph{Adaptive Control of Thought -- Rational (ACT-R)} \cite{anderson_integrated_2004} is a widely employed cognitive architecture. It is a modular production rule system with a special architecture of the working memory that operates on data stored as so-called \emph{chunks}, i.e. the unit of knowledge in the human brain. It has a well-defined psychological theory, however, its computational system is not described formally leading to implementations that are full of technical artifacts as claimed in \cite{albrecht_2014a,stewart_deconstructing_2007}, for instance. This impedes formal reasoning about the underlying languages and the programmed models. It makes it hard to compare different implementation variants of the languages. Furthermore, it complicates verifying properties of the models. These issues call for a formal semantics of cognitive modeling languages together with proper analysis techniques.

In this paper, we describe a very abstract formulation of the operational semantics of ACT-R, called the \emph{very abstract semantics}. This formalization of the fundamental parts of ACT-R is the basis for the specification of concrete implementations of the theory of the architecture, as well as for analysis of cognitive models. This very abstract semantics, extending and adapting the work of \cite{albrecht_2014a}, captures all possible kinds of ACT-R implementations. This can be too abstract to formally reason about actual computational models meaningfully. 

Therefore, we give a concrete instance of this very abstract semantics that is suitable for analysis and implementation -- the \emph{abstract semantics}. It still abstracts from a variety of technical details like conflict resolution, times and latency that strongly depend on the actual implementation or even configuration of ACT-R, but includes the typical matching process of rules and some common actions that can be extended. This paves the way for computational analysis of actual cognitive models. 

Eventually, we construct a sound and complete \emph{translation} of ACT-R models in abstract semantics to Constraint Handling Rules (CHR). The abstract semantics is therefore directly executable via its CHR representation making it, to the best of our knowledge, the first formal operational semantics of ACT-R that is suitable for analysis and execution.

The paper is a revised and extended version of our prior work in \cite{gall_ppdp2015} and \cite{gall_ruleml2016}. It has the following contributions:
\begin{enumerate}
 \item A very abstract operational semantics of ACT-R (c.f. section~\ref{sec:very_abstract_semantics}), \label{enum:very_abstract}
 \item an instance of this semantics for analysis (abstract semantics, c.f. section~\ref{sec:abstract_semantics_variant}), \label{enum:abstract}
 \item a translation of ACT-R models to the declarative programming language Constraint Handling Rules (CHR) (c.f. section~\ref{sec:translation}) and \label{enum:translation}
 \item a soundness and completeness proof of the translation w.r.t. the operational semantics (c.f. section~\ref{sec:soundness_completeness}). \label{enum:translation:proofs} 
\end{enumerate}
The formulations of the semantics (\ref{enum:very_abstract} and~\ref{enum:abstract}) have been improved compared to \cite{gall_ppdp2015} and the translation (\ref{enum:translation}) has been revised substantially compared to \cite{gall_ruleml2016} making it suitable for the proofs. The soundness and completeness proofs (\ref{enum:translation:proofs}) have not been published before and are fundamentally new.

\emph{Constraint Handling Rules (CHR)} \cite{fru_chr_book_2009} has a formally defined operational semantics as well as a declarative semantics with corresponding soundness and completeness results. There are many theoretical and practical results and tools for analysis of CHR programs \cite{fru_chr_book_2009}. Due to its strengths in formal program analysis and its strong relation to first order \cite{fru_chr_book_2009} and linear logic \cite{fru_chr_book_2009,betz_linlog_2005}, it has been used as a lingua franca that embeds many rule-based approaches \cite{fru_chr_book_2009} like term rewriting systems \cite{raiser_towards_trs_2008}, graph transformation systems \cite{raiser_gts_chr_2007,raiser_analysis_gts_2011} and business rules \cite{martin_fages_business_rules_cpdc07}. Such embeddings have been used successfully to make the analysis results of CHR available to other approaches. The sound and complete embedding of ACT-R in CHR enables the use of these results and tools to formally reason about cognitive models.

The paper is structured as follows:
We first give a short introduction to ACT-R in section~\ref{sec:preliminaries:actr}. The formal definition of the very abstract semantics is given in section~\ref{sec:very_abstract_semantics} and the  the abstract semantics is defined as an instance in section~\ref{sec:abstract_semantics_variant}. The translation scheme of abstract ACT-R models to CHR programs is described in section~\ref{sec:translation} and proven sound and complete w.r.t. the operational semantics in section~\ref{sec:soundness_completeness}. Related work is discussed in section~\ref{sec:related_work} with a detailed comparison to prior work.

\section{Description of ACT-R}
\label{sec:preliminaries:actr}

In this section, we describe ACT-R informally. For a detailed introduction to the theory, we refer to \cite{AndersonLe98,anderson_integrated_2004,taatgen_modeling_2006,anderson_how_2007}. Adaptive Control of Thought -- Rational (ACT-R) is a popular cognitive architecture that is used in many cognitive models to describe and explain human cognition. There have been applications in language learning models \cite{taatgen_why_2002} or in improving human computer interaction by the predictions of a cognitive model \cite{byrne2001act}. The components of the ACT-R architecture even have been mapped to brain regions \cite[chapter 2]{anderson_how_2007}.

Using a cognitive architecture like ACT-R simplifies the modeling process, since well-investigated psychological results have been assembled to a unified theory about fundamental parts of human cognition. In the best-case, such an architecture constrains modeling to only plausible cognitive models \cite{taatgen_modeling_2006}. Computational cognitive models are described clearly and unambiguously since they are executed by a computer producing detailed simulations of human behavior \cite{sun_introduction_2008}. By performing the same experiments on humans and the implemented cognitive models, the resulting data can be compared and models can be validated.

\subsection{Overview of the ACT-R Architecture}
\label{sec:preliminaries:actr:overview}

The ACT-R theory is built around a modular production rule system operating on data elements called \emph{chunks}. A chunk is a structure consisting of a name and a set of labeled slots that are connected to other chunks. The slots of a chunk are determined by its \emph{type}. The names of the chunks are only for internal reference -- the information represented by a network of chunks comes from the connections. For instance, there could be chunks representing the cognitive concepts of numbers 1, 2, \dots{} By chunks with slots \emph{number} and \emph{successor} we can connect the individual numbers to an ordered sequence describing the concept of natural numbers. This is illustrated in figure~\ref{fig:natural_numbers}. 

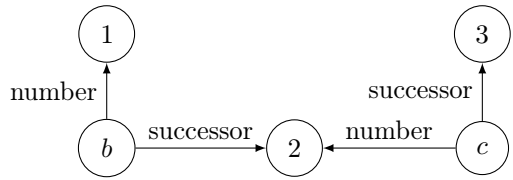
\begin{figure}[htb]
\centering
\tikzstyle{slot} = [draw, -latex]   

\begin{tikzpicture}[node distance = 2.5cm, auto]
 \node[chunk] (b) {$b$}; 
 \node[chunk, above of=b, node distance=1.5cm] (b-first) {1}; 
 \node[chunk, right of=b] (b-second) {2}; 

 \node[chunk, right of=b-second] (c) {$c$}; 
 \node[chunk, above of=c, node distance=1.5cm] (c-second) {3}; 

 \path[slot] (b) -- node {number} (b-first);
 \path[slot] (b) -- node[above] {successor} (b-second);
 \path[slot] (c) -- node[above] {number} (b-second);
 \path[slot] (c) -- node {successor} (c-second);
\end{tikzpicture}
\caption{Two count facts with names $b$ and $c$ that model the counting chain 1, 2, 3.}
\label{fig:natural_numbers}
\end{figure}

As shown in figure~\ref{fig:act_r_architecture}, ACT-R consists of modules. The \emph{goal module} keeps track of the current (sub-) goal of the cognitive model. The \emph{declarative module} contains \emph{declarative knowledge}, i.e. factual knowledge that is represented by a network of chunks. There are also modules for interaction with the environment like the \emph{visual} and the \emph{manual} module. The first perceives the visual field whereas the latter controls the hands of the cognitive agent. Each module is connected to a set of \emph{buffers} that can hold at most one chunk at a time.

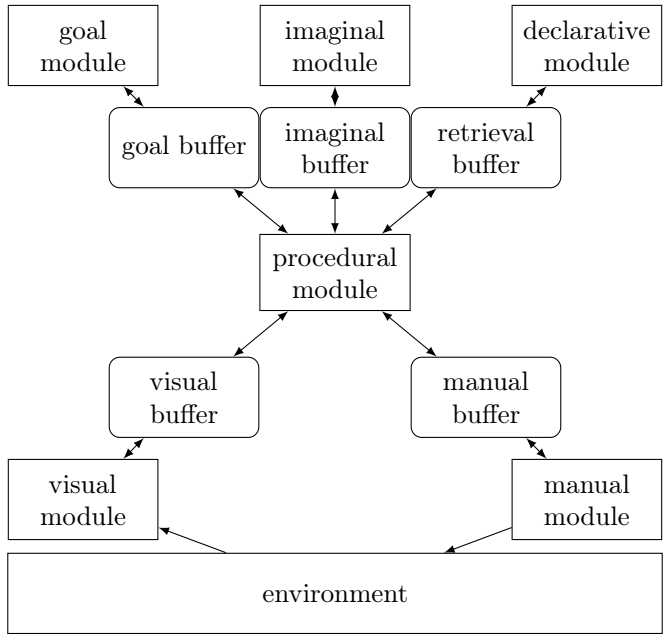
\begin{figure}[htb]
\begin{center}
\tikzstyle{module} = [rectangle, draw,
    text width=5em, text centered, minimum height=3em]
\tikzstyle{buffer} = [rectangle, draw, 
    text width=5em, text centered, rounded corners, minimum height=3em]
\tikzstyle{line} = [draw, -latex]
\tikzstyle{doubleline} = [draw, latex-latex]

\usetikzlibrary{calc}
    
\begin{tikzpicture}[node distance = 1.9cm, auto]
 \matrix[row sep=6mm] (modules) {
 \node[module] (goalmod) {goal module}; 
 \node[buffer, below right of=goalmod] (goalbuf) {goal buffer}; 
&  \node[module] (imaginalmod) {imaginal module}; 
 \draw let
    \p1=(goalbuf.center),
    \p2=(imaginalmod.south) in
 node[buffer] at (\x2,\y1) (imaginalbuf) {imaginal buffer}; &
 \node[module] (declmod) {declarative module}; 
 \node[buffer, below left of=declmod] (declbuf) {retrieval buffer}; 
\\
 & \node[module, minimum height=1cm, minimum width=2cm] (procmod) {procedural module}; \\

 \node[module] (visualmod) {visual module}; 
 \node[buffer, above right of=visualmod] (visualbuf) {visual buffer}; 
& &
 \node[module] (manualmod) {manual module}; 
 \node[buffer, above left of=manualmod] (manualbuf) {manual buffer}; 
\\
};

\draw let
    \p1=(procmod.center),
    \p2=(visualmod.south),
     \p3=(visualmod.west),
    \p4=(manualmod.east) in
 node[module,   minimum width=\x4-\x3, text width=6em] at (\x1,\y2-7mm) (environment) {environment}; 

  \path [doubleline] (goalmod) -- (goalbuf);
\path [doubleline] (goalbuf) -- (procmod);
 \path [doubleline] (declmod) -- (declbuf);
\path [doubleline] (declbuf) -- (procmod);
 \path [doubleline] (imaginalmod) -- (imaginalbuf);
\path [doubleline] (imaginalbuf) -- (procmod);
  \path [doubleline] (visualmod) -- (visualbuf);
\path [doubleline] (visualbuf) -- (procmod);
 \path [doubleline] (manualmod) -- (manualbuf);
\path [doubleline] (manualbuf) -- (procmod);
\path [line] (environment) -- (visualmod);
 \path [line] (manualmod) -- (environment);
\end{tikzpicture}
\end{center}
\caption{Modular architecture of ACT-R. This illustration is inspired by \cite{taatgen_modeling_2006} and \cite{anderson_integrated_2004}.}
\label{fig:act_r_architecture}
\end{figure}

The heart of the system is the \emph{procedural module} that contains the production rules controlling the cognitive process. It only has access to a part of the declarative knowledge: the chunks that are in the buffers. A production rule matches the content of the buffers and -- if applicable -- executes its actions. There are three types of actions: 
\begin{description}
 \item[Modifications] overwrite information in a subset of the slots of a buffer, i.e. they change the connections of a chunk.
 \item[Requests] ask a module to put new information into its buffer. The request is encoded in form of a chunk. The implementation of the module defines how it reacts on a request. For instance, there are modules that only accept chunks of a certain form like the manual module that only accepts chunks that encode a movement command for the hand according to a predefined set of actions. 
 
 Nevertheless, all modules share the same interface for requests: The module receives the arguments of the request encoded as a chunk and puts its result in the requested buffer. For instance, a request to the declarative module is stated as a partial chunk and the result is a chunk from the declarative knowledge (the fact base) that matches the chunk from the request.
 \item[Clearings] remove the chunk from a buffer.
\end{description}

The system described so far is the so-called \emph{symbolic level} of ACT-R. It is similar to standard production rule systems operating on symbols (of a certain form) and matching rules that interact with buffers and modules. However, to simulate the human mind, a notion of timing, latency, priorities etc. are needed. In ACT-R, those concepts are subsumed in the \emph{sub-symbolic level}. It augments the symbolic structure of the system by additional information to simulate the previously mentioned concepts.

Therefore, ACT-R has a progressing simulation time. Certain actions can take some time that depends on the information from the sub-symbolic level. For instance, chunks are mapped to an \emph{activation level} that determines how long it takes the declarative module to retrieve it. Activation levels also resolve conflicts between chunks that match the same request. The value of the activation level depends on the usage of the chunk in the model (inter alia): Chunks that have been retrieved recently and often have a high activation level. Hence, the activation level changes with the simulation time. This can be used to model learning and forgetting of declarative knowledge. Similarly to the activation level of chunks, production rules have a \emph{utility} that also depends on the context and the success of a production rule in prior applications. Conflicts between applicable rules are resolved by their utilities which serve as dynamic, learned rule priorities.

\subsection{Syntax}
\label{sec:syntax}

We use a simplified syntax of ACT-R that we have introduced in \cite{gall_lopstr2014}. It is based on sets of logical terms instead of the concatenation of syntactical elements. This enables an easier access to the syntactical parts. Our syntax can be transformed directly to the original ACT-R syntax and vice-versa.

The syntax of ACT-R is defined over two possibly infinite, disjoint sets of (constant) symbols $\consts$ and variable symbols $\vars$. An ACT-R model consists of a set of types $\types$ with type definitions and a set of rules $\rules$. A production rule has the form $L \Rightarrow R$ where $L$ is a finite set of buffer tests and queries. A buffer test is a first-order term of the form $\test{b}{t}{P}$ where the buffer $b \in \consts$ and $P \subseteq \consts \times (\consts \cup \vars)$ is a set of slot-value pairs $(s,v)$ where $s \in \consts$ and $v \in \consts \cup \vars$. This means that only the values in the slot-value pairs can consist of both constants and variables. The right-hand side $R \subseteq \actions$ of a rule is a finite set of actions where $\actions = \{ a(b,t,P) \enspace | \enspace a \in A, b \in \consts, t \in \consts \mbox{ and } P \subseteq \consts \times (\consts \cup \vars) \}$. I.e. an action is a term of the form $a(b,t,P)$ where the functor $a$ of the action is in $A$, the set of action symbols, the first argument $b$ is a constant (denoting a buffer), the second argument is a constant $t$ denoting a type, and the last argument is a set of slot-value pairs, i.e. a pair of a constant and a constant or variable. Usually, the action symbols are defined as $A := \{=,+,-\}$ for modifications, requests and clearings respectively. Only one action per buffer is allowed, i.e. if $a(b,t,P) \in R$ and $a'(b',t',P') \in R$, then $b \neq b'$ \cite{actr_reference}.

We define the function $\fun{vars}$ that maps an arbitrary set of terms to its set of variables in $\vars$. For a production rule $L \Rightarrow R$ the following must hold: $\fvars{R} \subseteq \fvars{L}$, i.e. no new variables must be introduced on the right-hand side of a rule. As we will see in the following sections about semantics, this restriction demands that all variables are bound on the left-hand side.

\subsection{Informal Operational Semantics}

In this section, we describe ACT-R's operational semantics informally. The production rule system constantly checks for matching rules and applies their actions to the buffers. This means that it tests the conditions on the left hand side with the contents of the buffers (which are chunks) and applies the actions on the right hand side, i.e. modifies individual slots, requests a new chunk from a module or clears a buffer.

The left hand side of a production rule consists of buffer tests -- that are terms $\test{b}{t}{P}$ with a buffer $b$, a type $t$ and a set of slot-value pairs $P$. The values of a slot-value pair can be either constants or variables. The test matches a buffer, if the chunk in the tested buffer $b$ has the specified type $t$ and all slot-value pairs in $P$ match the values of the chunk in $b$. Thereby, variables of the rule are bound to the actual values of the chunk. Values of a chunk in the buffers are always ground. This is ensured by the previously mentioned condition in the syntax of a rule that the right hand side of a rule does not introduce new variables (see section~\ref{sec:syntax}). Hence the chunks in the buffers stay ground.

If there is more than one matching rule, a conflict resolution mechanism that depends on the sub-symbolic layer chooses one rule that is applied. After a rule has been selected, it takes a certain time (usually 50\,ms) for the rule to fire. I.e. actions are applied after this delay. During that time the procedural module is blocked and no rule can match.

The right hand side consists of actions $a(b,t,P)$, where $a \in A$ is an action symbol, $b$ is a constant denoting a buffer and $P$ is again a set of slot-value pairs. We have already explained the three types of actions (modifications, requests and clearings) roughly. In more detail, a modification overwrites only the slots specified in $P$ with the values from $P$. A request clears the requested buffer and asks a module for a new chunk. It can take some time specified by the module (and often depending on sub-symbolic values) until the request is processed and the chunk is available. During that time, other rules still can fire, i.e. requests are executed in parallel. However, a module can only process one request for a buffer at the same time. Buffer clearings simply remove the chunk from a buffer. In the following, we disregard clearings in our definitions since they are easy to add.

We now give an example rule and informally explain its behavior.
\begin{example}[production rule]
\label{ex:production_rule}
We want to model the counting process of a little child that has just learned how to count from one to ten. We use the natural number chunks described in section~\ref{sec:preliminaries:actr:overview} as declarative knowledge. Furthermore, we have a goal chunk of another type $g$ that memorizes the current number in a \emph{current} slot. We now define a production rule, that increments the number in the counting process (and call this rule \emph{inc}). We denote variables with capital letters in our examples. The left-hand side of the rule \emph{inc} consists of two tests:
\begin{itemize}
 \item $\test{\var{goal}}{g}{\{ (\var{current}, X) \}}$ and
 \item $\test{\var{retrieval}}{\var{succ}}{\{ (\var{number}, X), (\var{successor}, Y) \}}$.
\end{itemize}
This means that the rule tests if in the goal buffer there is a chunk of type $g$ that has some number $X$ (which is a variable) in the \emph{current} slot. If this number $X$ is also in the \emph{number} slot of the chunk in the retrieval buffer, the test succeeds and the variable $Y$ is bound to the value in the \emph{successor} slot. The actions of the rule are:
\begin{itemize}
 \item $\actionmod{\var{goal}}{g}{\{ (\var{current}, Y) \}}$ and
 \item $\actionreq{\var{retrieval}}{\var{succ}}{\{ (\var{number}, Y) \}}$.
\end{itemize}
The first action modifies the chunk in the goal. A modification cannot change the type, that is why we just add an anonymous variable denoted by the underscore symbol in the type specification. The \emph{current} slot of the goal chunk is adjusted to the successor number $Y$ and the declarative module is asked for a chunk of type $\var{succ}$ with $Y$ in its \emph{number} slot. This is called a retrieval request. After a certain amount of time, the declarative module will put a chunk with $Y$ in its \emph{number} and $Y+1$ in its \emph{successor} slot into the retrieval buffer and the rule can be applied again.
\end{example}

\section{Very Abstract Operational Semantics} 
\label{sec:very_abstract_semantics}

Our goal is to define an operational semantics that captures as many ACT-R implementations as possible and leaves room for extensions and modifications.  Hence, we first give them a common theoretical foundation that is based on the formalization according to \cite{albrecht_2014a}
-- the very abstract operational semantics. It describes the fundamental concepts of a production rule system that operates on buffers and chunks like ACT-R. This work extends the definition from \cite{albrecht_2014a}. We compare our work with the work from \cite{albrecht_2014a} in section~\ref{sec:related_work:albrecht}. Later on, in section~\ref{sec:abstract_semantics_variant}, we define an instance of this very abstract semantics that is suitable for analysis of cognitive models.

An \emph{ACT-R architecture} is a concrete instantiation of the very abstract semantics and defines general parts of the system that are left open by the very abstract semantics like the set of possible actions $A$, the effect of such an action or the selection process. In contrast to that, an \emph{ACT-R model} defines model-specific instantiations of parts like the set of types $\types$ and the set of rules $\rules$. Table~\ref{tab:very_abstract_semantics:arch_vs_mod} summarizes what is defined by the architecture and the model.

\subsection{Chunk Stores}

As described before in section~\ref{sec:preliminaries:actr}, ACT-R operates on a network of typed chunks that we call a \emph{chunk store}. Therefore, we first define the notion of types:
\begin{definition}[chunk types]
A \emph{typing function} $\tau : \types \rightarrow 2^\consts$ maps each type from the set $\types \subseteq \consts$ to a finite set of allowed slot names. Every type set $\types$ must contain a special type $\chunk \in \types$ with $\tau(\chunk) = \emptyset$. 
\end{definition}
A chunk store is defined over a set of types and a typing function. We abstract from chunk names as they do not add any information to the system. In fact, chunks are defined as unique, immutable entities with a type and connections to other chunks:
\begin{definition}[chunk store]
\label{def:chunk_store}
A \emph{chunk store} $\Delta$ is a multi-set of tuples $(t,\val)$ where $t \in \types$ is a chunk type and $\val : \tau(t) \rightarrow \Delta$ is a function that maps each slot of the chunk (determined by the type $t$) to another chunk. 

Every chunk store $\Delta$ must contain a chunk $\nil \in \Delta$ that is defined as $\nil := (\chunk,\emptyset)$. Each chunk store $\Delta$ has a bijective \emph{identifier function} $\id_\Delta : \Delta \rightarrow \consts$ that maps each chunk of the multi-set a unique identifier. The inverse of $\id$ is defined as follows:
\begin{equation*}
\id^{-1}_\Delta(x) := \begin{cases}
                    c    & \text{if $\id_\Delta(c) = x$,}\\
                    \nil & \text{otherwise.}\\
                   \end{cases}
\end{equation*}

For a chunk $c = (t,\val)$, the following functions are defined:
\begin{itemize}
 \item $\fun{type}(c) = t$ and
 \item $\fun{slots}(c) = \val$.
\end{itemize} 
\end{definition}
The typing function $\tau$ maps a type $t$ from the set of type names $\types$ to a set of allowed slots, hence the function $\val$ of chunk $c$ has the slots of $c$ as domain. 

Note that two chunks are only considered equivalent, if they have the same chunk identifier, type and value functions (in that case they are the indistinguishable in the multi-set and therefore treated as one and the same chunk). Hence, a chunk store can contain multiple elements with the same values that still are unique entities representing different concepts. We will see this in the following example: We model our well-known example from figure~\ref{fig:natural_numbers} as a chunk store.
\begin{example}[chunk store of natural numbers]
\label{ex:chunk_store_nat}
The chunk store from figure~\ref{fig:natural_numbers} can be modeled as follows. Note that in the examples, we write $c {::}_\Delta (t,\val)$ to denote that the specified chunk on the right hand side has identifier $c$. When it is clear from the context, we only write ${::}$ instead of ${::}_\Delta$. We also use the chunk identifiers to define the connections in the slots (the $\val$ functions). 
\begin{itemize}
 \item The set of types is $\types_{\ref{ex:chunk_store_nat}} = \{ number, succ \}$.
 \item The typing function $\tau_{\ref{ex:chunk_store_nat}} : \types \to 2^\consts$ is defined as $\tau_{\ref{ex:chunk_store_nat}}(\var{number}) = \emptyset$ and $\tau_{\ref{ex:chunk_store_nat}}(\var{succ}) = \{ \var{number}, \var{successor} \}$.
 \item We have the following chunks in our store $\Delta_{\ref{ex:chunk_store_nat}}$:
 \begin{itemize}
  \item the unique entities with identifiers $1, 2, 3$ that are defined as $(\var{number},\emptyset)$,
  \item $b :: (\var{succ},\val_b)$ with $\val_b(s) = \begin{cases}
                                                     1 & \text{if $s = \var{number}$}\\
                                                     2 & \text{if $s = \var{successor}$}
                                                    \end{cases}$
  \item $c :: (\var{succ},\val_c)$ with $\val_b(s) = \begin{cases}
                                                     2 & \text{if $s = \var{number}$}\\
                                                     3 & \text{if $s = \var{successor}$}
                                                    \end{cases}$
 \end{itemize}

\end{itemize}
\end{example}

From the definition of a chunk store $\Delta$, we can derive a graph $(\Delta,E)$ where for each slot-value pair $\val(s) = d$ of a chunk $c = (t,\val) \in \Delta$ there is an edge $(c,d) \in E$ with label $s$. In the graphical representation, we can label vortices with the chunk identifiers. We can apply this to our example~\ref{ex:chunk_store_nat} to derive the graph illustrated in figure~\ref{fig:natural_numbers} on page~\pageref{fig:natural_numbers}.

Sometimes we want to build chunk stores that only have a few chunks in them that refer to chunks in other chunk stores in their slots. We call this concept a \emph{partial chunk store}.
\begin{definition}[partial chunk store]
\label{def:partial_chunk_store}
 A \emph{partial chunk store} with reference to a chunk store $\Delta$, denoted as $\Delta^{\circledcirc}$, is a multi-set of tuples $(t,\val)$ where $t \in \types$ is a chunk type and $\val : \tau(t) \rightarrow \Delta \uplus \Delta^{\circledcirc}$ is a function that maps each slot of the chunk (determined by the type $t$) to another chunk from chunk store $\Delta \uplus \Delta^{\circledcirc}$. Every chunk in the partial chunk store $\Delta^{\circledcirc}$ has a unique identifier that is disjoint from the identifiers in $\Delta$. The function $\id_{\Delta^{\circledcirc}} : \Delta^{\circledcirc} \rightarrow \consts$ returns the chunk identifier for each chunk in $\Delta^{\circledcirc}$.

The set of all partial chunk stores that refer to a chunk store $\Delta$ is denoted as $\DDeltap{\Delta}$.
\end{definition}

\begin{example}[partial chunk store]
\label{ex:partial_chunk_store}
Let $\Delta_{\ref{ex:chunk_store_nat}}$ be the chunk store from example~\ref{ex:chunk_store_nat}. We define $\Delta_{\ref{ex:partial_chunk_store}}$ as a partial chunk store that refers to $\Delta_{\ref{ex:chunk_store_nat}}$. It contains the chunk $x ::_{\Delta_{\ref{ex:partial_chunk_store}}} (\mathit{succ},\val_x)$  with the following slots:
\begin{equation*}
\val_x := \begin{cases}
           2 & \text{if $s = \var{number}$}\\
           3 & \text{if $s = \var{successor}$}
          \end{cases}
\end{equation*}

\end{example}

We define an operation that merges two (partial) chunk stores. In an abstract way it can be considered as a special multi-set union that merges two elements of a chunk store, if they have the same chunk identifiers. However, since there are many different implementations of ACT-R, we do not want to limit our formulation to this special type of multi-set union, but define a more general operator $\chunkmerge$. For the general understanding of the paper and the proofs it is sufficient to think of it as multi-set union that maintains uniqueness of chunk identifiers.
\begin{definition}[chunk merging]
\label{def:chunkmerge}
Let $\DDelta$ be the set of all chunk stores, $\Delta \in \DDelta$ a chunk store and $\DDeltap{\Delta}$ the set of all partial chunk stores that refer to $\Delta$. Then $\chunkmerge : (\DDelta \cup \DDeltap{\Delta}) \times (\DDelta \cup \DDeltap{\Delta}) \rightarrow (\DDelta \cup \DDeltap{\Delta})$ the \emph{chunk merging operator}. In the following, $\Delta \subseteq \Delta'$ denotes that every element of $\Delta$ is an element of $\Delta'$ with preservative chunk identifiers. 

We require the following properties for $\chunkmerge$. For all chunk stores $\Delta \in \DDelta$ and all (partial) chunk stores $\Delta', \Delta'', \Delta''' \in \DDelta \cup \DDeltap{\Delta}$:
\begin{enumerate}
 \item $(\Delta' \chunkmerge \Delta'') \chunkmerge \Delta''' = \Delta' \chunkmerge (\Delta'' \chunkmerge \Delta''')$, i.e. $\chunkmerge$ is \emph{associative}.  \label{def:chunkmerge:assoc}
 \item $\emptyset$ is the \emph{neutral element}. 
 \item For all $c' \in \Delta'$ and $c'' \in \Delta''$ with $\id_{\Delta'}(c') = \id_{\Delta''}(c'')$, it must hold that $c' = c''$, i.e. both chunks have same types and value functions.
 \item For all (partial) chunk stores $X \in \DDelta \cup \DDeltap{\Delta}:$ If $X \subseteq \Delta'$ and $X \subseteq \Delta''$, then there exist  (partial) chunk stores $Y$ and $Z$ such that $\Delta' \chunkmerge \Delta'' = X \chunkmerge Y \chunkmerge Z, \Delta' = X \chunkmerge Y$ and $\Delta'' = X \chunkmerge Z$. This is a special definition of \emph{idempotency} of $\chunkmerge$ for elements with the same chunk identifiers, i.e. chunks in both stores that have the same chunk identifiers, types and value functions are absorbed in the merge product. \label{def:chunkmerge:idempot}
 
 There are the following interesting special cases:
 \begin{itemize}
  \item $X = \Delta'$ leads to $\Delta' \chunkmerge \Delta'' = \Delta' \chunkmerge Z$ with $\Delta'' = \Delta' \chunkmerge Z$, i.e. the elements of $\Delta''$ that also appear in $\Delta'$ are absorbed in the merge product.
  \item $\Delta' = \Delta''$ leads to $\Delta' \chunkmerge \Delta'' = \Delta'$ with $X = \Delta' = \Delta''$ and $Y = Z = \emptyset$.
 \end{itemize}

 \item If $c' \in \Delta'$, then  $c' \in \Delta' \chunkmerge \Delta''$ and $\id_{\Delta'}(c') = \id_{\Delta' \chunkmerge \Delta''}(c')$, i.e. $\Delta' \subseteq \Delta' \chunkmerge \Delta''$ with preservative chunk identifiers. 
 \item There is a mapping $\fun{map}_{\Delta',\Delta''} : \Delta' \uplus \Delta'' \rightarrow \Delta' \chunkmerge \Delta''$ that maps chunks from the original chunk stores to the merged chunk store. It must hold that $\fun{map}_{\Delta',\Delta''}(c) = c$ if $c \in \Delta'$, i.e. the chunks from $\Delta'$ remain untouched by the merging. \label{def:chunkmerge:subset}
\end{enumerate}
From the axioms it is clear that $(\DDelta \cup \DDeltap{\Delta}, \chunkmerge)$ is a monoid. 
\end{definition}
We now give two examples for possible implementations of $\chunkmerge$.

The simplest chunk merging operator would be the multi-set union that maintains uniqueness of chunk identifiers, i.e. if two chunks have the same identifiers, types and value functions, only one version will be kept in the merged store. The  chunk identifiers are defined as 
\begin{equation*}
\id_{\Delta \chunkmerge \Delta'}(c) := \begin{cases}
                                 \id_{\Delta}(c) & \text{if $c \in \Delta$}\\
                                 \id_{\Delta'}(c) & \text{otherwise.}
                                  \end{cases} 
\end{equation*}
      
For instance, let $\Delta := \{ c_1, c_2 \}$ and $\Delta' := \{ c_2, c_3 \}$ with $\id_\Delta(c_i) = i$ for $i = 1, 2 $ and $\id_{\Delta'}(c_i) = i$ for $i = 2, 3$. We assume that $c_i = c_i$ for $i = 1, 2, 3$, i.e. those chunks are equivalent and therefore have same types and value functions. Then $\Delta \chunkmerge \Delta' = \{ c_1, c_2, c_3 \}$. Due to the same identifiers, types and value functions of the two appearances of $c_2$, the merged store only keeps one version of $c_2$ due to the idempotency axiom.

For chunk stores whose chunk identifiers have been renamed apart, this definition of $\chunkmerge$ is equivalent to real multi-set union. Therefore, $\Delta \chunkmerge \Delta' = \Delta \uplus \Delta'$ for all chunk stores that have disjoint chunk identifiers.
                                                                                                                                                         
                                                                                                                                                                                      However, in most implementations chunks that have the same structure are merged to one chunk, i.e. if $c := (t,\val) \in \Delta_p^{(\Delta)}$ and $c' := (t,\val) \in \Delta_{p'}^{(\Delta)}$, then $c$ and $c'$ would be merged to $c$ in $\Delta_p^{(\Delta)} \chunkmerge \Delta_{p'}^{(\Delta)}$. I.e., only $c$ would be kept in the merged store. The mapping function would return $\fun{map}_{\Delta_{p}^{(\Delta)},\Delta_{p'}^{(\Delta)}}(c') = c$.

\subsection{States}

We first define the individual parts of an ACT-R state. The notion of a \emph{cognitive state} defines which chunks are currently in which buffer and therefore visible to the production system that can only match chunks in buffers.
\begin{definition}[cognitive state]
\label{def:cognitive_state}
A \emph{cognitive state} $\gamma$ is a function $\buffers \rightarrow \Delta \times \mathbb{R}^+_0$ that maps each buffer to a chunk and a delay. The set of cognitive states is denoted as $\cogstates$, whereas $\cogstatespart$ denotes the set of \emph{partial cognitive states}, i.e. cognitive states that are partial functions and do not necessarily map each buffer to a chunk. We define the following functions to access the individual parts of a cognitive state $\gamma$: If $\gamma(b) = (c,d)$ for an arbitrary buffer $b$, then
\begin{itemize}
 \item $\fun{chunk}(\gamma(b)) = c$ and
 \item $\fun{delay}(\gamma(b)) = d$.
\end{itemize}
\end{definition}
The delay decides at which point in time the chunk in the buffer is available to the production system. A delay $d > 0$ indicates that the chunk is not yet available to the production system. This implements delays of the processing of requests.  

ACT-R adds a sub-symbolic level to the symbolic concepts that have been defined so far and that distinguish it from other production rule systems. To gather information from the sub-symbolic layer, we add the concept of (sub-symbolic) additional information that is needed to calculate sub-symbolic values. This information can be altered by an abstract function as we will see in section~\ref{sec:very_abstract_semantics:opsem}. The information will be expressed as multi-sets or conjunctions of predicates from first-order logic. The additional information is also used to manage data used in ACT-R's modules.

Additionally, modules other than the procedural module hold their data in the additional information.

We now define ACT-R states as follows:
\begin{definition}[very abstract state]
\label{def:unified_state}
A \emph{very abstract state} is a tuple $\langle \Delta; 
\gamma; \addinfo; t \rangle$ where 
$\gamma$ is a cognitive state in the sense of definition~\ref{def:cognitive_state}, $\addinfo$ is a multi-set of ground, atomic first order predicates (called \emph{additional information}), $t \in \mathbb{R}^+_0$ is a time. The state space is denoted with $\states_\veryabstract$.
\end{definition}
Note that a very abstract state cannot contain variables from $\vars$, but is only compound from terms, sets and functions over constants from $\consts$: The chunk store $\Delta$ contains a type that is denoted by a constant and a valuation function that connects slot names (constants) to other elements from $\Delta$, the cognitive state $\gamma$ connects buffer names (constants) with chunks from $\Delta$ and a delay in $\mathbb{R}^+_0$, the additional information $\addinfo$ is a multi-set of ground, atomic predicates and the time is also a number. The set of allowed predicates for additional information is denoted as $\addinfos$. The additional information holds data of ACT-R's modules as well as sub-symbolic information.

We continue our running example by defining a very abstract state with one of the chunks defined in example~\ref{ex:chunk_store_nat}.
\begin{example}[ACT-R states]
\label{ex:very_abstract_state}
We want to model the counting process of a little child that has learned the sequence of the natural numbers from one to ten as declarative facts and can retrieve those facts from declarative memory.
Therefore, we add a chunk of type $g$ with a \emph{current} slot that memorizes the current number in the counting process.

The following state has a chunk of type $g$ in the goal buffer that has the current number $1$. The retrieval buffer is currently retrieving the chunk $b$ with number $1$ and successor $2$. The retrieval is finished in one second as denoted by the delay. Figure~\ref{fig:ex:very_abstract_state} illustrates the state. The formal definition is:
\begin{itemize}
 \item $\types_{\ref{ex:very_abstract_state}} = \types_{\ref{ex:chunk_store_nat}} \cup \{ g \}$ where $\types_{\ref{ex:chunk_store_nat}}$ is the set of types from example~\ref{ex:chunk_store_nat}.
 \item $\tau_{\ref{ex:very_abstract_state}}(t) = \begin{cases}
                                                \{ \var{current} \}        & \text{if $t = g$}\\
                                             \tau_{\ref{ex:chunk_store_nat}}(t) & \text{otherwise.}
                                            
                                            \end{cases}$
 \item $\Delta_{\ref{ex:very_abstract_state}} = \Delta_{\ref{ex:chunk_store_nat}} \uplus \{ (g,\fun{val}_\mathit{goal}) \}$ where $\fun{val}_\mathit{goal}(\var{current}) = 1$.
 \item $\sigma_0 = \langle \Delta_{\ref{ex:very_abstract_state}} ; \gamma_0 ; \constr{true} ; 0\rangle$
 \item $\gamma_0(\var{goal}) = ((g,\fun{val}_\mathit{goal}),0)$ 
 \item $\gamma_0(\var{retrieval}) = (b,1)$ where $b$ is defined as in example~\ref{ex:chunk_store_nat}.
\end{itemize}
\end{example}

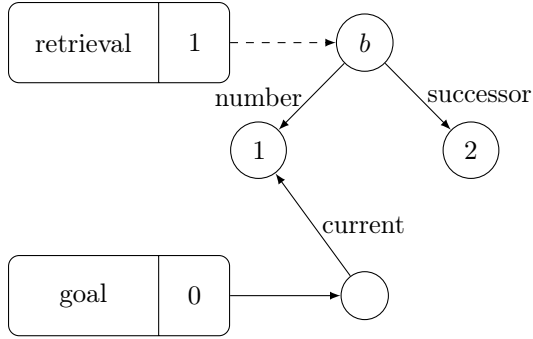
\begin{figure}[htb]
\centering

\begin{tikzpicture}[node distance=2cm,auto]
\matrix[row sep = 1cm] {

\node[buffer] (b1) {retrieval \nodepart[text width=2em]{second} 1 }; &

\node[chunk] (c) {$b$};
\node[chunk, below left of=c] (v1) {1}; 
\node[chunk, below right of=c] (v2) {2}; 

\\
\node[buffer] (b2) {goal \nodepart[text width=2em]{second} 0}; &

\node[chunk] (c2) {};

\\
};

\path[line,dashed] (b1) --  (c);
\path[line] (c) -- node[left] {number} (v1);
\path[line] (c) -- node[right] {successor} (v2);
\path[line] (b2) -- (c2);
\path[line] (c2) -- node[right] {current} (v1);
\end{tikzpicture}
\caption{Visual representation of the very abstract state defined in example~\ref{ex:very_abstract_state}. The dashed arrow signifies that the chunk in the retrieval buffer is not yet visible (as indicated by the delay right of the buffer's name).}
\label{fig:ex:very_abstract_state}
\end{figure}

\subsection{Operational Semantics}
\label{sec:very_abstract_semantics:opsem}

We now define the state transition system of the very abstract semantics. As in every production rule system, we first define how matching rules are chosen. Therefore, we introduce a selection function $S$ that is defined by the architecture and maps a state to a set of matching rules and the variable bindings implied by the application of the rule.
\begin{definition}[selection function]
Let $\Theta(\vars,\consts)$ be the set of possible substitutions over variables $\vars$ and constants $\consts$. A \emph{selection function} is a function  $S : \states_\veryabstract \rightarrow 2^{\rules \times \Theta(\vars,\consts)}$ that maps a state to a set of pairs $(r,\theta)$ where $r \in \rules$ is a production rule and $\theta \in \Theta(\vars,\consts)$ is a substitution of variables from $\vars$ with constants from $\consts$, such that all variables from the rule $r$ are substituted, i.e. $\fun{dom}(\theta) = \fun{vars}(r)$. 
\end{definition}
For actual implementations of ACT-R, the result of $S$ is usually restricted to sets with zero or one element, but for abstract definitions there can also be more than one rule. The function $S$ usually defines a notion of matching and makes sure that only rules can fire that match visible information in the buffers, i.e. chunks that are not delayed by a time greater than zero.

To define the modification of a state by a transition, we define interpretation functions of actions that determine the possible effects of an action.
\begin{definition}[interpretation of actions] 
\label{def:interpretation_action}
An \emph{interpretation of an action} is a function  $I : \actions \times \states_\veryabstract \rightarrow 2^{\DDeltap{\Delta} \times \cogstatespart \times \addinfos}$. The following conditions must hold: $(\Delta^*, \gamma^*, \addinfo^*) \in I(\alpha,\sigma)$ if
\begin{enumerate}
 \item $I(\alpha,\sigma) \neq \emptyset$, i.e. the interpretation of an action has at least one effect,
 \item the resulting chunk store is $\Delta^*$ is a partial chunk store that refers to $\Delta$ and whose chunk identifiers are disjoint from $\Delta$,
 \item the co-domain of $\gamma^*$ is $\Delta^* \times \mathbb{R}_0^+$, i.e. the cognitive state can only refer to chunks in the resulting chunk store $\Delta^*$, and
 \item if the action $\alpha$ has the buffer $b$ in its scope, i.e. $\alpha := (b,t,p)$, then the resulting partial cognitive state $\gamma^*$ has only $b$ in its domain, i.e. $\fun{dom}(\gamma^*) = \{ b \}$. 
\end{enumerate}

\end{definition}
An interpretation maps each state and action of the form $a(b,t,P)$ -- where $a \in A$ is an action symbol, $b \in \consts$ a constant denoting a buffer, $t \in \consts$ a type, and $P \subseteq \consts \times (\consts \cup \vars)$ is a set of slot-value pairs -- to a tuple $(\Delta^*,\gamma^*,\addinfo^*)$. Thereby, $\Delta^*$ is a partial chunk store that refers to $\Delta$, $\gamma^*$ is a partial cognitive state, i.e. a partial function that assigns only the buffer $b$ from the action to a chunk and a delay. The partial cognitive state $\gamma^*$ will be taken in the operational semantics to overwrite the changed buffer contents, i.e. it contains the new contents of the changed buffers. Analogously, the additional information $\addinfo^*$ defines the additions to the sub-symbolic level induced by the action.

Note that the interpretation of an action can return more than one possible effect. This is used in the abstract semantics where due to the lack of sub-symbolic information all possible effects have to be considered. For example, the declarative module can find more than one chunk matching the retrieval request. Usually, by comparing activation levels of chunks, one chunk will be returned. However, in the abstract semantics all matching chunks are possible. In the refined semantics, we restrict the selection to one possible effect as proposed by the ACT-R reference manual \cite{actr_reference}.

\begin{example}[interpretation of an action]
\label{ex:neutral_effect}
In this example we define an \emph{neutral effect}. We will see later that if this effect is applied to a state, the state does not change modulo time.

Let $\alpha := a(b,t,p)$ be our action that produces the neutral effect and $\sigma := \langle \Delta ; \gamma ; \addinfo ; t \rangle$ an ACT-R state. We define
\begin{equation*}
I(\alpha,\sigma) := \{ (\{ c \},\gamma',\constr{true}) \}
\end{equation*}
where $\gamma'(b) := (c,d)$ for $c := (t,\val)$ if $\gamma(b) = (t,\val)$ and some $d \in \mathbb{R}_0^+$ and undefined for all other inputs.

Intuitively, the action returns a chunk store with only one chunk $c$ that has the same type and value as the chunk $\gamma(b)$ that has been in buffer $b$ in the state $\sigma$. The resulting partial cognitive state just links $b$ to this new chunk $c$. No additional information ($\constr{true}$) is added.

This is a valid interpretation, since $I(\alpha,\sigma) \neq \emptyset$, $\Delta^*$ is a valid partial chunk store that refers to $\Delta$. Therefore $\val : \tau(t) \rightarrow \Delta$ is a valid value function. The co-domain of $\gamma'$ is $\Delta^* \times \mathbb{R_0^+}$ since the co-domain of $\gamma$ is $\Delta \times \mathbb{R}_0^+$ and by definition~\ref{def:chunkmerge}, $\Delta \subseteq \Delta \chunkmerge \Delta'$. Additionally, $\fun{dom}(\gamma') = \{ b \}$.
\end{example}

To combine interpretations of all actions of a rule, we first define how two interpretations can be combined. Therefore, we introduce the following set operator, that combines two sets of sets:
\begin{definition}[combination operator $\comb$ for effects]
Let $e, f$ be effects of two actions, i.e. results of an interpretation function of an action in a state $\sigma$ with chunk store $\Delta$, i.e. $e \in I(\alpha',\sigma)$ and $f \in I(\alpha'',\sigma)$. Let $e := (\Delta',\gamma',\addinfo')$ and $f := (\Delta'',\gamma'',\addinfo'')$ where $\Delta', \Delta''$ are partial chunk stores with disjoint chunk identifiers that refer to $\Delta$, $\gamma' : \buffers' \rightarrow \Delta \chunkmerge \Delta' \times \mathbb{R}_0^+, \gamma'' : \buffers'' \rightarrow \Delta \chunkmerge \Delta'' \times \mathbb{R}_0^+$ are partial cognitive states with disjoint domains, i.e. $\buffers', \buffers'' \subseteq \buffers$ and $\buffers' \cap \buffers'' = \emptyset$ and $\addinfo$ is a conjunction of first-order predicates. Then the \emph{combination} of the effects $e$ and $f$ w.r.t. a chunk store $\Delta$ is defined as
\begin{equation*}
e \comb f := (\Delta' \chunkmerge \Delta'',\gamma, \addinfo' \land \addinfo'')
\end{equation*}
where $\gamma : \buffers' \cup \buffers'' \rightarrow (\Delta' \chunkmerge \Delta'') \times \mathbb{R}_0^+$ with 
\begin{equation*}
 \gamma(b) := \begin{cases}
               \fun{map}_{\Delta',\Delta''}(\gamma'(b)) & \text{if $b \in \fun{dom}(\gamma')$,}\\
               \fun{map}_{\Delta',\Delta''}(\gamma''(b)) & \text{if $b \in \fun{dom}(\gamma'')$.}
              \end{cases}
\end{equation*}
\end{definition}
The intuition behind this definition is that two effects of an action, i.e. two triples of chunk store, cognitive state and additional information, are merged to one effect that combines them. Hence, we get a merged partial cognitive state that has the combined buffer-chunk mappings of the two original cognitive states. This is possible, since the domains of the partial cognitive states are required to be disjoint. 

The partial chunk stores are merged to one partial chunk store referring to the same total chunk store. Note that from the definition of $\chunkmerge$, the chunk store of the first effect is a subset of the merged store. However, the second chunk store might have lost some members. The mapping function assigns every chunk from the second store one from the merged store.

The merged additional information is a conjunction of the additional information from both effects.

The combination is well-defined: $\Delta' \chunkmerge \Delta''$ exists since $\Delta'$ and $\Delta''$ are partial chunk stores that refer to the same chunk store $\Delta$. Additionally, they have disjoint identifiers because they are results of an interpretation function (see definition~\ref{def:interpretation_action}) and therefore their merging cannot fail due to different chunks with same identifiers.  

The cognitive state is valid, since it has the combined domain of $\gamma'$ and $\gamma''$ and just maps the chunks from those original partial cognitive states to their versions in the merged chunk store by $\fun{map}$, i.e. the co-domain of $\gamma$ is just the merge product of the co-domains of $\gamma'$ and $\gamma''$ and keeps their connections of buffers to chunks. 


The definition of combinations of effects can be lifted to sets of effects by the following definition.
Let $E$ and $F$ be two sets of effects in some state $\sigma$ with chunk store $\Delta$, then their combination is defined as all possible pairwise combination of their elements:
\begin{equation*}
E \comb F := \{ e \comb f ~|~ e \in E \land f \in F \}.
\end{equation*}
Since the interpretation of an action is possibly non-deterministic, i.e. might have more than one effect triple, the combination of such sets of effect triples is a set that combines each effect from the first set with each effect from the second set. This leads to a set of combined effects from which the transition system will be able to choose one non-deterministically. However, since every effect set is required to have at least one effect, it the same applies for their combination.


We now define the interpretation function $I : \rules \to 2^{\DDelta \times \cogstatespart \times \addinfos}$ that maps a rule to all its possible effects (i.e. chunk store, cognitive state and additional information).
\begin{definition}[interpretation of rules]
\label{def:interpretation_rule}
In a state $\sigma := \langle \Delta ; \gamma ; \addinfo ; t \rangle$, a rule $r := L \Rightarrow R$ is interpreted by an interpretation function $I : \rules \times \states_\veryabstract \to 2^{\DDeltap{\Delta} \times \cogstatespart \times \addinfos}$ that is defined as follows: $I(r,\sigma)$ applies the function $\applyrule_r$ to all tuples in the result set when combining the individual actions of the rule:
\begin{itemize}
 \item $\applyrule : \DDeltap{\Delta} \times \cogstatespart \times \addinfos \to \DDeltap{\Delta} \times \cogstatespart \times \addinfos$ is a function that applies some more effects at the end of the rule application and is defined by the architecture.
 \item For all actions $\alpha \in R$, the resulting chunk stores are disjoint (i.e. chunk identifiers are renamed apart).
 \item The interpretation has the result 
 \begin{equation*}
  I(r,\sigma) = \applyrule_r(\bigcomb_{\alpha \in R} I(\alpha,\sigma))
 \end{equation*}
 where the combination operator $\comb$ refers to $\Delta$ and the $\applyrule$ function is applied to each member of the combination set.
 Hence, all possible effects of the rules are combined and each of the resulting partial cognitive states is then modified by the $\applyrule$ function that is defined by the architecture. 
\end{itemize}
\end{definition}
The $\applyrule$ function can apply additional changes to the state that are not directly defined by its actions. For instance, it can change some sub-symbolic values that depend on the rule application like the utility of the rule itself. Note that by definition of the ACT-R syntax it is ensured that each of the $\gamma_\mathrm{part}$ in the combination of the individual actions is still a function, since only one action per buffer is allowed as defined in section~\ref{sec:syntax}. 

It is important to mention, that there are two types of non-determinism in the interpretation of a rule:
\begin{enumerate}
 \item The first non-determinism comes from the non-deterministic nature of interpretations of an action. Each action can lead to different results (depending on their definition). This is why all interpretation functions have power sets of effects as co-domain.
 \item The second type of non-determinism comes from the definition of the combination operator that merges chunk stores by using the chunk merging operator $\chunkmerge$. Since $\chunkmerge$ is not required to be commutative, the result of the merged chunk stores may vary. This leads to possibly differing chunk identifiers. It is possible to abstract from this kind of non-determinism by introducing the concept of (graph) isomorphism on chunk stores.
\end{enumerate}

We now define the operational semantics as the state transition system $(\states_\veryabstract,\utrans)$:
\begin{definition}[very abstract operational semantics]
\label{def:unified_operational_semantics}
The transition relation $\utrans : \states_\veryabstract \times \states_\veryabstract$ in the \emph{very abstract operational semantics} of ACT-R is defined as follows:
\begin{description}
 \item[Apply]
 For a rule $r$ the following transitions are possible:
 \begin{equation*}
  \frac{ 
    (r,\theta) \in S(\sigma), (\Delta^*,\gamma^*,\addinfo^*) \in I(r\theta,\sigma)
  }{
    \sigma := \langle \Delta ; \gamma; \addinfo; t \rangle \utrans^r_\mathbf{apply} \langle \Delta \chunkmerge \Delta^* ; \gamma'; \addinfo \land \addinfo^*; t' \rangle
  }
 \end{equation*}
 where 
 \begin{itemize}
  \item $\gamma' : \buffers \rightarrow \Delta \chunkmerge \Delta^*$,\\
   $\gamma'(b) := \begin{cases}
                       (\fun{map}_{\Delta,\Delta^*}(c),d) & \text{if } \gamma^*(b) = (c,d) \text{ is defined}\\
                       (\fun{map}_{\Delta,\Delta^*}(c),d \ominus \delayrule)          & \text{otherwise, if $\gamma(b) = (c,d)$,}
                      \end{cases}
$
  \item $x \ominus y := \begin{cases}
                    x - y & \text{if $x > y$}\\
                    0     & \text{otherwise}
                   \end{cases}
$ for two numbers $x, y \in \mathbb{R}_0^+$,
 and 
  \item $t' = t + \delayrule$ for a delay $\delayrule \in \mathbb{R}^+_0$ defined by the concrete instantiation of ACT-R.
 \end{itemize}
 When applying the rule, the resulting partial chunk store $\Delta^*$ is merged with the chunk store $\Delta$ from the state. Hence, $\Delta \chunkmerge \Delta^* = \Delta \chunkmerge \Delta^*_{\alpha_1} \chunkmerge_\Delta \dots  \chunkmerge_\Delta \Delta^*_{\alpha_n}$ for all actions $\alpha_i$ on the right-hand  side of the rule. Note that $\Delta \subseteq \Delta \chunkmerge \Delta^*$, i.e. all chunks in $\Delta$ also appear in the merged chunk store with preservative chunk identifiers by definition~\ref{def:chunkmerge} of the chunk merging.
 
 The partial cognitive state that comes from the interpretation of the rule replaces all positions in the original cognitive state where it is defined, otherwise the original cognitive state remains untouched. Note that for all buffers $b \in \buffers$ and $b \notin \fun{dom}(\gamma^*)$ with $\gamma(b) = (c,d)$ we can also write $\gamma'(b) := (c,d \ominus \delayrule)$ instead of $\gamma'(b) := (\fun{map}_{\Delta,\Delta^*}(c),d \ominus \delayrule)$ since the chunk merging guarantees that $\Delta \subseteq \Delta \chunkmerge \Delta^*$ with preservative identifiers as mentioned before.
 
 The delays are taken from the partial cognitive state $\gamma^*$ or are reduced by a constant amount that models progression of time.

 When it is clear from the context, we just use $\utrans^r$ to denote that the transition applies rule $r$.
 
 \item[No Rule]
 \begin{equation*}
  \frac{ 
    C(\sigma)
  }{ 
    \sigma := \langle \Delta ; \gamma ; \addinfo ; t \rangle \utrans_\mathbf{no} \langle \Delta ; \gamma'; \addinfo; \newtime(\sigma) \rangle
  }
 \end{equation*}
 where
 \begin{itemize}
  \item $C \subseteq \states_\veryabstract$ is a side condition in form of a logical predicate,
  \item $\applynorule : \states_\veryabstract \to \cogstatespart$ a function that describes how the cognitive state should be transformed,
  \item $\newtime : \states_\veryabstract \to \mathbb{R}^+_0$ a function that describes the time adjustment in dependency of the current state, and
  \item $\gamma'(b) := \begin{cases}
                        \applynorule(\sigma)(b) & \text{if defined}\\
                        \gamma(b)                       & \text{otherwise}
                      \end{cases}$ is the updated state.

 \end{itemize}
\end{description}
\end{definition}

An \emph{apply} transition applies a rule that satisfies the conditions of the selection function $S$ by overwriting the cognitive state $\gamma$ with the result from the interpretations of the actions of rule $r$. Thereby, one possible combination of all effects of the actions is considered. Note that the transition is also possible for all other combinations. Only the buffers with a new chunk are overwritten, the others keep their contents. The same applies for parameters: They keep their value except for those where $\addinfo^*$ defines a new value. Additionally, the rule application can take a certain time $\delayrule$ that is defined by the architecture. Time is forwarded by $\delayrule$, i.e. the time in the state is incremented by $\delayrule$ and the delays in the cognitive state that determine when a chunk becomes visible to the system are decremented by $\delayrule$ (with a minimal delay of 0).

The \emph{no rule} transition defines what happens if there is no rule applicable, but there are still effects of e.g. requests that can be applied. This means that there are buffers $b \in \buffers$ with $\gamma(b) = (c,d)$ and $d > 0$, i.e. information that is not visible to the production rule system. In that case there are no possible transitions in the original semantics. We generalized this case in our definition of the \emph{no rule} transition that allows state transitions without rule applications. It ensures that if a side condition $C(\sigma)$ defined by the ACT-R instantiation, the cognitive state is updated according to the function $\applynorule$ and the current time of the system is set to a specified time $\newtime(\sigma)$. Both functions are also defined by the concrete architecture. This makes new information visible to the production system and hence new rules might fire. In typical ACT-R implementations, the side condition $C(\sigma)$ is that  $S(\sigma) = \emptyset$, i.e. that no rule is applicable, and $\newtime(\sigma) := t + d^*$ where $\sigma$ has the time component $t$ and $d^*$ is the minimum delay in the cognitive state of $\sigma$. This means that time is forwarded to the minimal delay in the cognitive state and makes for instance pending requests visible to the production rule system. It can be interpreted like if the production rule system waits with the next rule application until there is new information present that leads to a rule matching the state. This behavior coincides with the specification from the ACT-R reference implementation \cite{actr_reference}. If no transition is applicable in a state $\sigma$, i.e. there is no matching rule and no invisible information in $\sigma$, then $\sigma$ is a final state and the computation stops.

The definition of our very abstract semantics leaves parts to be defined by the actual architecture and the model. Table~\ref{tab:very_abstract_semantics:arch_vs_mod} summarizes what has to be defined by an architecture and a model.

\begin{table}
\caption{Parameters of the very abstract semantics that must be defined by the architecture or the cognitive model respectively.}
\centering
\begin{tabular}{l m{0.45\linewidth} ll}
\toprule
\multicolumn{2}{c}{Architecture} & \multicolumn{2}{c}{Model} \\
\midrule
$\consts$  & set of constants    & $\types$ & set of types\\
$\vars$ & set of variables       & $\tau$ & typing function\\
$\buffers$ & set of buffers      & $\rules$ & set of rules \\ 
$A$ & set of action symbols      & $\sigma_S$ & start state\\
$\delayrule$ & rule delay  \\
$\addinfos$ & allowed additional information\\
$\chunkmerge$ & chunk merging operator\\
$S$ & rule selection function  \\
$I$ & interpretation functions  \\
$\applynorule$ & transformation of cognitive state after no rule transition \\
$\newtime$ & progress of time after no rule transition    \\
\bottomrule
\end{tabular}
\label{tab:very_abstract_semantics:arch_vs_mod}
\end{table}

We now show that the neutral effect from example~\ref{ex:neutral_effect} is a neutral element of the rule application except for the time component.
\begin{example}[neutral element of rule applications]
 Let $\sigma := \langle \Delta ; \gamma ; \addinfo ; t \rangle$ be an ACT-R state and $r := L \Rightarrow \{ \alpha \}$ an ACT-R rule with only one action $\alpha := a(b^*,t^*,P^*)$. Let $\gamma(b^*) := (t,\val)$. 
 
 We define $e := I(\alpha,\sigma) := \{ (\{ c^* \},\gamma^*,\constr{true}) \}$ to be the only effect of $\alpha$ where $c^* := (t,\val)$ and $\gamma^*(b^*) := (c^*,d^*)$ for some $d^* \in \mathbb{R}_0^+$ and undefined for all other inputs. Then
\begin{equation*}
I(r,\sigma) := \{ (\{ c^* \},\gamma^*, \constr{true}) \}
\end{equation*}
where $\gamma^* : \buffers' \cup \buffers'' \rightarrow \{ c^* \} \times \mathbb{R}_0^+$ with 

Let $\sigma \utrans^r \sigma'$. Then $\sigma' := \langle \Delta \chunkmerge \{ c^* \} ; \gamma' ; \addinfo \land \constr{true} ; t'\rangle.$ The follow-up cognitive state $\gamma'$ is
\begin{equation*}
  \gamma'(b) := \begin{cases}
                       (\fun{map}_{\Delta,\{ c^* \}}(c),d) & \text{if } \gamma^*(b) = (c,d) \text{ is defined}\\
                       (\fun{map}_{\Delta,\{ c^* \}}(c),d \ominus \delayrule)          & \text{otherwise, if $\gamma(b) = (c,d)$.}
                      \end{cases}
\end{equation*}
This can be reduced to
\begin{equation*}
  \gamma'(b) :=  \begin{cases}
                       (c^*,d) & \text{if }  b = b^*\\
                       (c,d \ominus \delayrule)          & \text{otherwise, if $\gamma(b) = (c,d)$,}
                      \end{cases}
\end{equation*}
Since $c^*$ has the same type and value function,  $\gamma'(b)$ can be considered equivalent to $\gamma(b)$ modulo delays for all $b \in \buffers$. Hence, $\sigma$ is equivalent to $\sigma'$ modulo delays and the time component.
\end{example}

\section{Abstract Semantics as Instance of the Very Abstract Semantics}
\label{sec:abstract_semantics_variant}
%

The \emph{abstract semantics} is defined as an instance of the very abstract semantics. It is suitable for the analysis of procedural core of cognitive models because it abstracts from timings and conflict resolution by leaving parts of the transition system to non-deterministic choices and still giving room for extensions and modifications. The idea is to define the minimal core of all implementations of ACT-R's procedural system disregarding parameter choices, timings, sub-symbolic information and module configuration. The abstract semantics captures all possible state transitions the procedural system can make.

\subsection{Definition of the Abstract Semantics}
\label{sec:abstract_semantics_variant:definition}

Since it is a central part of the procedural system of ACT-R, we first define the notion of matchings:
\begin{definition}[matching]
\label{def:matching}
A buffer test $\beta := \test{b}{\var{ct}}{P}$ for a buffer $b \in \buffers$ testing for a type $\var{t} \in \types$ and slot-value pairs $P \subseteq \consts \times (\consts \cup \vars)$ \emph{matches} a state $\sigma := \langle \Delta ; \gamma ; \addinfo ; t \rangle$, written $\beta \matches \sigma$, if and only if there is a substitution $\theta$ such that $\gamma(b) = ((t,\fun{val}), 0)$ and for all $(s,v) \in P : \fun{id}_\Delta(\fun{val}(s)) = v\theta$. We define $\fun{Bindings}(\beta,\sigma) := \theta$ as the function that returns the smallest substitution that satisfies the matching $\beta \matches \sigma$.

This definition can be extended to rules: A rule $r := L \Rightarrow R$ matches a state $\sigma$, written as $r \matches \sigma$, if and only if for all buffer tests $t \in L$ match $\sigma$. The function $\fun{Bindings}(r,\sigma)$ returns the smallest substitution that satisfies the matching $r \matches \sigma$.
\end{definition}
This means that a buffer test matches a state, if the tested buffer contains a chunk of the tested type and all slot tests hold in the state, i.e. the variables in the test can be substituted by values consistently such that they match the values from the state. The values in the test are denoted by the identifiers of the chunks. Note that a test can only match chunks in the cognitive state that are visible to the system, i.e. whose delay is zero. A test cannot match chunks with a delay greater than zero.

We give the architectural parameters that are left open in the very abstract semantics:
\begin{description}
 \item[States] 
 We set the time in every state to $t := 0$ (or any other constant) because abstract states are not timed. Hence, each abstract state is a tuple $\langle \Delta ; \gamma ; \addinfo ; 0 \rangle$ where $\gamma \in \cogstates$ is a cognitive state. We sometimes project an abstract state $\langle \Delta ; \gamma ; \addinfo ; 0 \rangle$ to $\langle \Delta ; \gamma ; \addinfo \rangle$ for the sake of brevity.
 \item[Selection Function] 
 The rule selection in the abstract semantics is simply defined as $S_{\mathit{abs}}(\sigma) := \{ (r, \fun{Bindings}(r,\sigma)) \enspace | \enspace r \in \rules \wedge r \matches \sigma \}$. Hence we select all matching rules in state $\sigma$ and replace the variables from the rules by their actual values from the matching, since in the state transition system the substitution $\theta$ is applied to the rule when calculating the effects. 
 \item[Effects] For a state $\sigma = \langle \Delta ; \gamma ; \addinfo ; 0 \rangle$ the interpretation function $I_{\mathit{abs}}$ for actions $A := \{ \mathtt{=}, \mathtt{+} \}$ in the abstract semantics is defined as follows:
 \begin{itemize}
  \item  $I_{\mathit{abs}}(=(b,t,P),\sigma) = \{(\Delta^*,\gamma^*, \constr{true})\}$ for modifications\\
  where
  \begin{itemize}
   \item $c := (t,\val_b)$,
   \item $\Delta^* := \{ c \}$ with $\id_{\Delta^*}(c) = i$ for a fresh id $i \in \consts$,
   \item $\gamma^*(b) := (c,0)$, and
   \item the new slot values are: 
   \begin{equation*}
   \val_b(s) := \begin{cases}
    \fun{id}^{-1}_\Delta(v)  & \text{if $(s,v) \in P$}\\
    \val_\gamma(s) & \text{otherwise.}
  \end{cases}
  \end{equation*}
  for  $\gamma(b) = ((t,\val_\gamma),d)$.
  
  This means that a modification creates a new chunk that modifies only the slots specified by $P$ and takes the remaining values from the chunk that has been in the buffer. Note that the type cannot be modified, since the resulting chunk always has the type derived from the chunk that has previously been in the buffer. Modifications are deterministic, i.e. that there is only one possible effect. 
  
  The slot-value function $\val_b$ is well-defined, since it appears in a partial chunk store that references $\Delta$, it has $\Delta$ as co-domain by definition~\ref{def:partial_chunk_store} of a partial chunk store. 
  
  If the action contains a slot-value pair $(s,v)$ that modifies $s$ to a chunk that was not existent in the original state $\sigma$, this chunk is not magically constructed, since we do not know its type or values. Instead, we map this slot to $\nil$. This comes from the definition of $\fun{id}^{-1}_\Delta$, which is $\nil$ for $v$, since the chunk referenced by $v$ does not exist in $\Delta$.  
  \end{itemize}
  \item $(\Delta^*,\gamma^*, \addinfo^*) \in I_{\mathit{abs}}(+(b,t,P),\sigma)$ for requests\\
  if
  \begin{itemize}
   \item $\fun{request}_b : \types \times 2^{\consts \times (\consts \cup \vars)} \times \addinfos \to 2^{\Delta \times \mathbb{R}^+_0 \times \addinfos}$ is a function defined by the architecture for each buffer. It calculates the set of possible answers for a request that is specified by a type and a set of slot value pairs. Possible answers are tuples $(c,d,\addinfo)$ of a chunk $c$, delay $d$ and parameter valuation  function $\addinfo$.
   \item For all $(c^*,d^*,\addinfo^*) \in \fun{request}_b(t,P,\addinfo)$ we set $\gamma^*(b) := \begin{cases}
                         (c^*,1) & \text{if $d^* > 0$}\\
                         (c^*,0) & \text{otherwise,}
                         \end{cases}$\\
        and $\Delta^* := \{ c^* \}$.
  \end{itemize}
  Note that the additional information in the result of the request is directly added to the result of the interpretation function to update the internal state of the requested module.
  \item The function $\applyrule$ from definition~\ref{def:interpretation_rule} that adds additional changes to the state when a rule is applied is defined as the identity function, i.e. no changes to the state are introduced by the rule application itself but only by its actions.
 \end{itemize}
 \item[Rule Application Delay] 
 The delay of a rule application is set to $\delayrule := 0$, since the abstract semantics does not care about timings. 
 
 \item[No Rule Transition] 
 In the \emph{no rule} transition, there are three parameters to be defined by the actual ACT-R instantiation: The side condition $C$, the state update function $\applynorule$ and the time adjustment function $\newtime$. We define them for a state $\sigma := \langle \gamma ; \addinfo ; t \rangle$ as follows:
 \begin{itemize}
  \item $\sigma \in C$ if and only if there is a $b^* \in \buffers$ such that $\gamma(b^*) = (c,d)$ with $d > 0$, i.e. there is a buffer with a chunk that is not visible to the system. Those are the cases where there is a pending request. This means that the \emph{no rule} transition is possible as soon as there is at least one pending request. We call the buffer of one such request $b^*$.
  \item $\left[\applynorule(\sigma)\right](b^*) := (c,0)$ if $\gamma(b^*) = (c,d)$ and $d > 0$ for one $b^* \in \buffers$. This means that one pending request is chosen to be applied (the one appearing in $C$). Since this is a rule scheme and $b^*$ can be chosen arbitrarily, the transition is possible for all assignments of $b^*$. This coincides with the original definition of our abstract semantics where one request is chosen from the set of pending requests.
  \item The function $\newtime$ that determines how the time is adjusted after a chunk has been made visible is defined as $\newtime(\sigma) := t$, i.e. the time is not adjusted.
 \end{itemize}
\end{description}

Table~\ref{tab:abstract_semantics:arch_vs_mod} shows the parameters of the abstract semantics that have to be defined by the architecture.

\begin{table}[hbt]
\caption{Parameters of the abstract semantics that must be defined by the architecture.}
\centering
\begin{tabular}{ll}
\toprule
\multicolumn{2}{c}{Architecture} \\
\midrule
$\consts$ & set of constants\\
$\vars$ & set of variables \\
$\buffers$ & set of buffers  \\
$A$ & set of action symbols \\
$\delayrule$ & rule delay  \\
$\addinfos$ & allowed additional information\\
$\chunkmerge$ & chunk merging operator\\
$\fun{request_b}$ & result of requests\\
\bottomrule
\end{tabular}
\label{tab:abstract_semantics:arch_vs_mod}
\end{table}

We summarize the transition scheme of the abstract semantics:
\begin{description}
 \item[Rule transition]
 
  \begin{equation*}
  \frac{ 
    r \matches \sigma \land \theta = \fun{Bindings}(r,\sigma) \land (\Delta^*,\gamma^*,\addinfo^*) \in I(r\theta,\sigma)
  }{
    \sigma := \langle \Delta ; \gamma; \addinfo\rangle \utrans^r_\mathbf{apply} \langle \Delta \chunkmerge \Delta^* ; \gamma'; \addinfo \land \addinfo^* \rangle
  }
 \end{equation*}
 where $\gamma' : \buffers \rightarrow \Delta \chunkmerge \Delta^*$,\\
   $\gamma'(b) := \begin{cases}
                       (\fun{map}_{\Delta,\Delta^*}(c),d) & \text{if } \gamma^*(b) = (c,d) \text{ is defined}\\
                       (\fun{map}_{\Delta,\Delta^*}(c),d)          & \text{otherwise, if $\gamma(b) = (c,d)$.}
                      \end{cases}
$

Again, since $\Delta \subseteq \Delta \chunkmerge \Delta^*$ with preservative chunk identifiers, we can also write 
\begin{equation*}
\gamma'(b) := \begin{cases}
                       (\fun{map}_{\Delta,\Delta^*}(c),d) & \text{if } \gamma^*(b) = (c,d) \text{ is defined}\\
                       (c,d)          & \text{otherwise, if $\gamma(b) = (c,d)$.}
                      \end{cases}
\end{equation*}

\item[No rule transition]
 \begin{equation*}
  \frac{ 
    \gamma(b^*) = (c^*,d^*) \land d^* > 0
  }{ 
    \sigma := \langle \Delta ; \gamma ; \addinfo ; t \rangle \utrans_\mathbf{no} \langle \Delta ; \gamma'; \addinfo \rangle
  }
 \end{equation*}
 where  $\gamma'(b) := \begin{cases}
                        (c^*,0)   & \text{if $b = b^*$}\\
                       \gamma(b)  & \text{otherwise.}
                      \end{cases}
$
\end{description}

We now extend our running example by a derivation in the abstract semantics:
\begin{example}[abstract semantics]
\label{ex:abstract_semantics}
We begin with the state $\sigma_0$ from example~\ref{ex:very_abstract_state}. It is visualized in figure~\ref{fig:ex:very_abstract_state}. Then, the following derivations are possible. 
\begin{align*}
\sigma_0 & \utrans^\mathit{no} \langle \gamma_1 ; \emptyset ; 0 \rangle\\
         & \utrans^\mathit{inc} \langle \gamma_2 ; \emptyset ; 0 \rangle =: \sigma_2
\end{align*}
where
\begin{itemize}
 \item $\gamma_1(\var{retrieval}) = (b,0)$ (and $\gamma_1(\var{goal}) = \gamma_0(\var{goal})$ as in $\sigma_0$),
 \item $\gamma_2(\var{retrieval}) = (c,1)$ and
 \item $\gamma_2(\var{goal}) = ((g,\fun{val}_{g_2},0)$ where $\fun{val}_{g_2}(\var{current}) = 2$
\end{itemize}
In $\sigma_0$ no rule is applicable, but there is a pending request whose result is not visible for the production system. Hence, we can apply the \emph{no rule} transition which makes the chunk $b$ visible. Then the rule \emph{inc} from example~\ref{ex:production_rule} is applicable. If we assume that $\fun{request}_\mathit{retrieval}(\var{succ},\{ (\var{number}, 2) \},\addinfo) = (c_\mathit{req}, 1, \emptyset)$ for all additional information $\addinfo$ where $c_\mathit{req} = (\var{succ},\{ (\var{number}, 2), (\var{successor}, 3) \})$, i.e. a chunk of type $\var{succ}$ with the number 2 in the $\var{number}$ slot and 3 in the $\var{successor}$ slot, we reach the state $\sigma_2$ that is illustrated in figure~\ref{fig:ex:abstract_semantics}. Note that in this state again the \emph{no rule} transition is possible.
\end{example}

\begin{figure}[htb]
\centering

\usetikzlibrary{shapes}

\begin{tikzpicture}[node distance=2cm,auto]
\matrix[row sep = 1cm] {
\node[buffer] (b1) {retrieval \nodepart[text width=2em]{second} 1 }; &
 \node[chunk] (c) {$c$};

\node[chunk, below left of=c] (v1) {2}; 
\node[chunk, below right of=c] (v2) {3}; 

\\
 \node[buffer] (b2) {goal \nodepart[text width=2em]{second} 0 }; &

\node[chunk] (c2) {};
\\
};

\path[line,dashed] (b1) --  (c);

\path[line] (c) -- node[left] {number} (v1);
\path[line] (c) -- node[right] {successor} (v2);
\path[line] (b2) -- (c2);
\path[line] (c2) -- node[right] {current} (v1);
\end{tikzpicture}
\caption{Visual representation of state $\sigma_2$ from example~\ref{ex:abstract_semantics}.}
\label{fig:ex:abstract_semantics}
\end{figure}
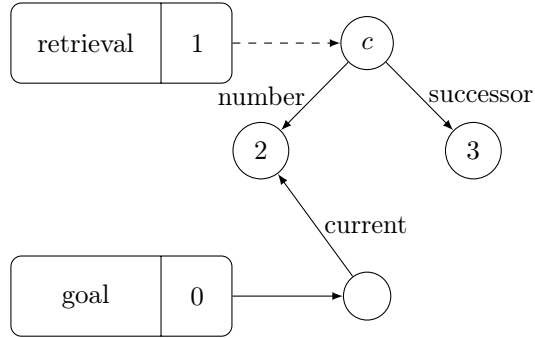

\section{Translation of ACT-R Models to Constraint Handling Rules}
\label{sec:translation}

In this section we show how to translate an ACT-R model to a Constraint Handling Rules (CHR) program. This is one of the main contributions of this paper. The translation is the first that matches the current operational semantics of ACT-R. The proof of the soundness and completeness of the translation w.r.t. the abstract operational semantics of ACT-R is a new contribution of this paper.

Therefore, we first give a quick introduction to CHR and its basic concepts we need for definitions and proofs in section~\ref{sec:constraint_handling_rules}. Then we introduce a normal form of ACT-R rules that simplifies the translation process and proofs in section~\ref{sec:set_normal_form}. In section~\ref{sec:translation_of_states} we show the translation of ACT-R states to CHR states and in section~\ref{sec:translation_of_rules} we define the translation scheme for rules. Finally, in section~\ref{sec:soundness_completeness} our translation of ACT-R models to CHR programs is proven sound and complete w.r.t. the abstract operational semantics of ACT-R.

\subsection{Constraint Handling Rules (CHR)}
\label{sec:constraint_handling_rules}

Before we define the translation scheme of ACT-R models to CHR program, we first recapitulate syntax and semantics of CHR briefly. For an extensive introduction to CHR, its semantics, analysis and applications, we refer to \cite{fru_chr_book_2009}. We use the latest definition of the state transition system of CHR that is based on state equivalence \cite{raiser_betz_fru_equivalence_revisited_chr09}. This syntax allows for more elegant proofs and has been proven to be equivalent to the canonical so-called very abstract semantics of CHR. The definitions from those canonical sources are now reproduced. We assume the reader to be familiar with the common concepts of first-order predicate logic like predicate symbols, function symbols, predicates, functions, constants, variables, terms and formulas.

The syntax of CHR is defined over a set of variables $\vars$, a set of function symbols (with arities) $\Phi$ and a set of predicate symbols with arities $\Pi$ that is disjointly composed of \emph{CHR constraint symbols} and \emph{built-in constraint symbols}. The set of constraint symbols contains at least the symbols $=/2$, $\constr{true}/0$ and $\constr{false}/0$.

For a constraint symbol $c/n \in \Pi$ and terms $t_1, \dots, t_n$ over $\vars$ and $\Phi$ for $1 \leq i \leq n$, $c(t_1,\dots,t_n)$ is called a \emph{CHR constraint}, if $c/n$ is a CHR constraint symbol or a \emph{built-in constraint} if $c$ is a built-in constraint respectively. We now define the notion of CHR states.
\begin{definition}[CHR state]
A \emph{CHR state} is a tuple $\langle \mathbb{G} ; \mathbb{B} ; \mathbb{V} \rangle$ where the \emph{goal} $\mathbb{G}$ is a multi-set of constraints, the \emph{built-in constraint store} $\mathbb{B}$ is a conjunction of built-in constraints and $\mathbb{V}$ is a set of \emph{global variables}.

All variables occurring in a state that are not global are called $\emph{local}$ and the local variables that are only used for built-in constraints are called $\emph{strictly local variables}$.
\end{definition}

CHR states can be modified by rules that together form a CHR program.
\begin{definition}[CHR program]
A \emph{CHR program} is a finite set of rules of the form 
\begin{equation*}
r \enspace @ \enspace H_k ~\backslash~ H_r \Leftrightarrow G ~|~ B_c , B_b
\end{equation*}
 where $r$ is an optional rule name, the heads $H_k$ and $H_r$ are multi-sets of CHR constraints, the guard $G$ is a conjunction of built-in constraints and the body is a multi-set of CHR constraints $B_c$ and a conjunction of built-in constraints $B_b$. Note that at most one of $H_k$ and $H_r$ can be empty. If $G$ is empty, it is interpreted as the built-in constraint $\constr{true}$. 
 
 If $H_k = \emptyset$, the rule is called a \emph{simplification rule} and we write
 \begin{equation*}
 r \enspace @ \enspace H_r \Leftrightarrow G ~|~ B_c , B_b.
 \end{equation*}
 Conversely, if $H_r = \emptyset$, the rule is called a \emph{propagation rule} and we write
 \begin{equation*}
 r \enspace @ \enspace H_k \Rightarrow G ~|~ B_c , B_b.
 \end{equation*}
\end{definition}

Informally, a rule is applicable, if the head matches constraints from the store $\mathbb{G}$ and the guard holds, i.e. is a consequence of the built-in constraints $\mathbb{B}$. In that case, the matching constraints from $H_k$ are kept in the store, the constraints matching $H_r$ are removed and the constraints from $B_c$, $B_b$ and $G$ are added.

To define the operational semantics formally, we first have to define state equivalence over CHR states following the work in \cite{raiser_betz_fru_equivalence_revisited_chr09}. Therefore, we assume a constraint theory $\mathcal{CT}$ for the interpretation of the built-in constraints in $\Pi$.
\begin{definition}[state equivalence of CHR states]
\label{def:chr_state_equiv}
 Equivalence between CHR states is the smallest equivalence relation $\equiv$ over CHR states that satisfies the following conditions:
 \begin{description}
  \item[Equality as substitution]
  
  \label{def:chr_state_equiv:substitution}
  \begin{equation*}
   \langle \mathbb{G} ; X {=} t \land \mathbb{B} ; \mathbb{V} \rangle \equiv \langle \mathbb{G}[X/t] ; X {=} t \land \mathbb{B} ; \mathbb{V} \rangle.
  \end{equation*}
  
  \item[Transformation of the constraint store] 
  
   \label{def:chr_state_equiv:equiv}
  
  If $\mathcal{CT} \models \exists \bar{s}.\mathbb{B} \leftrightarrow \exists \bar{s}'.\mathbb{B}'$ where $\bar{s}, \bar{s}'$ are the strictly local variables of $\mathbb{B}, \mathbb{B}'$, respectively, then:
  \begin{equation*}
   \langle \mathbb{G} ; \mathbb{B} ; \mathbb{V} \rangle \equiv \langle \mathbb{G} ; \mathbb{B}' ; \mathbb{V} \rangle.
  \end{equation*}

  \item[Omission of non-occurring global variables]
  
   \begin{equation*}
   \langle \mathbb{G} ; \mathbb{B} ; \{X\} \cup \mathbb{V} \rangle \equiv \langle \mathbb{G} ; \mathbb{B} ; \mathbb{V} \rangle.  
  \end{equation*}
  
  \item[Equivalence of failed states]
  
   \begin{equation*}
   \langle \mathbb{G} ; \constr{false} ; \mathbb{V} \rangle \equiv \langle \mathbb{G}' ; \constr{false} ; \mathbb{V}' \rangle.
  \end{equation*}
 \end{description}

\end{definition}

The operational semantics is now defined by the following transition scheme over equivalence classes over CHR states i.e. $[\sigma] ::= \{ \sigma' ~|~ \sigma' \equiv \sigma \}$
\begin{definition}[operational semantics of CHR]
\label{def:chr_operational_semantics}
For a CHR program the state transition system over CHR states and the rule transition relation $\mapsto$ is defined as the following transition scheme:
\begin{equation*}
\frac
 {
   r \enspace @ \enspace H_k ~\backslash~ H_r \Leftrightarrow G ~|~ B_c , B_b
 }
 {
 [\langle H_k \uplus H_r \uplus \mathbb{G} ; G \land \mathbb{B} ; \mathbb{V} \rangle ] \mapsto^r [ \langle H_k \uplus B_c \uplus \mathbb{G} ; G \land B_b \land \mathbb{B} ; \mathbb{V} \rangle ]
 }
\end{equation*}
Thereby, we assume that $r$ is a variant of a rule in the program such that its local variables are disjoint from the variables occurring in the representative of the pre-transition state. 

We may just write $\mapsto$ instead of $\mapsto^r$ if the rule $r$ is clear from the context.
\end{definition}

\subsection{Set Normal Form}
\label{sec:set_normal_form}

To simplify the translation scheme, we assume the ACT-R production rules to be in \emph{set normal form}. The idea is that each buffer test contains each slot exactly once. In the general ACT-R syntax it is possible to specify more than one slot-value pair for each slot or none at all. If we assume set-normal form, we can reduce the cases to consider in the translation and soundness and completeness proofs.

We now define the set-normal form formally and show that every ACT-R production rule can be transformed to set-normal form with same semantics.
\begin{definition}[set normal form]
An ACT-R rule $r := L \Rightarrow R$ is in \emph{set normal form}, if and only if for all buffer tests $\test{b}{t}{P} \in L$ and $s \in \tau(t)$ there is exactly one $v \in \vars \cup \consts$ such that $(s,v) \in L$. 
\end{definition}

\begin{theorem}[set normal form]
For all ACT-R models with rules $\rules$, there is a set of rules $\Sigma'$ with the following properties: For all rules $r \in \Sigma$ with $r := L \Rightarrow R$ that are applicable in at least one state there is a rule $r' \in \Sigma'$ with $r' := L' \Rightarrow R$ in set-normal form such that for all ACT-R states $\sigma, \sigma'$ it holds that $\sigma \utrans^r \sigma'$ if and only if $\sigma \utrans^{r'} \sigma'$ (in the abstract semantics).
\end{theorem}

We now give the construction of $\Sigma'$. All rules in $\Sigma$ that are not applicable in any state do not appear in $\Sigma'$. For all remaining rules $r \in \Sigma$, we first transform $r$ to a rule $r'$ that is in set-normal form. For every buffer test $\test{b}{t}{P} \in L$ there is a test $\test{b}{t}{P'\theta} \in L'$ for a substitution $\theta$. Thereby, $P'$ has the following slot-value pairs:
\begin{itemize}
 \item For all $s \in \tau(t)$ where there is exactly one $v \in \vars \cup \consts$ such that $(s,v) \in P$, $(s,v) \in P'$.
 \item For all $s \in \tau(t)$ where there is no $v \in \vars \cup \consts$ such that $(s,v) \in P$, there is a slot-value pair $(s,V) \in P'$ for a fresh variable $V \in \vars$.
 \item For all $s \in \tau(t)$ where there is are $v_1, \dots, v_n \in \vars \cup \consts$ such that $(s,v_i) \in P$ for $1 \leq i \leq n$ we produce a substitution $\theta_b := \{ v_1/v, \dots, v_n \}$ and add a slot-value pair $(s,v) \in P'$ for a $v \in \vars \cup \consts$. There are the following cases:
 \begin{enumerate}
  \item $v_1, \dots, v_n \in \vars:$ Then $v \in \vars$ is a fresh variable. Intuitively, we just introduce a new variable and replace all other variables by this new variable.
  \item $v_1, \dots, v_n \in \consts:$ Then $v = v_1$. Intuitively, since $r$ is applicable in at least one state, $v_1 = v_2 = \dots = v_n$. The slot-value pairs are redundant and therefore we only add one of them.
  \item $v_1, \dots, v_n \in \vars \cup \consts:$ W.l.o.g., $v_1 \in \consts$. Then $v = v_1$. 
 \end{enumerate}
\end{itemize}
We define the substitution $\theta$ as the composition of all substitutions $\theta_b$. The rule $r'$ is obviously in set normal form.

The proof idea is that if $r \matches \sigma$ all variables or constants that appear in the same slot test in the same buffer have the same values. If there is a buffer test $\test{b}{t}{P}$ with two slot-value pairs $(s,v) \in P$ and $(s,v')$, then $v\theta = v'\theta$ for $\theta = \fun{Bindings}(r,\sigma)$. Hence, we can replace all occurrences of $v$ and $v'$ by $v$ and remove $(s,v')$ from $P$. This is exactly what our construction of the set-normal form does.


\subsection{Notational Aspects}
\label{sec:notational_aspects}

We use the following symbols:
\begin{itemize}
 \item Lists are denoted by enumerating their elements, e.g. $[a, b, c]$ for the list with elements $a, b$ and $c$. $[]$ denotes the empty list. We use $[H|T]$ for a list with head element $H$ and tail list $T$. We also use the following notation for list comprehension: $[ x  : x \in M \land p(x)]$ is the list with all elements $x$ from a (multi-)set $M$ that satisfy $p(x)$.
 \item The list concatenation is denoted with $++$.
 \item For a function $f : A \rightarrow B$ with finite domain $A$, we define $\llbracket f \rrbracket := [(a,b) : a \in A \land b \in B \land f(a) = b]$ sorted by an order on $A$ and $B$, i.e. the (sorted) enumerative list notation of the function $f$. It can be understood as the list representation of the relational representation of $f$ as a set of tuples $f \subseteq A \times B$. 
\end{itemize}

\subsection{Translation of States}
\label{sec:translation_of_states}

To translate an ACT-R state to CHR, we have to define translations for the individual components of such a state.
\begin{definition}[translation of chunk stores]
\label{def:translation_chunk_store}
A (partial) chunk store $\Delta$ can be translated to a first order term as follows:
\begin{equation*}
[\constr{chunk}(\id_\Delta(c),t,\llbracket \val \rrbracket) : c \in \Delta \land c = (t,\val)]
\end{equation*}
Thereby, $\llbracket f \rrbracket$ denotes the explicit relational notation of the function $f$ as a sorted list of tuples and $[x : p(x)]$ is the list comprehension as defined in section~\ref{sec:notational_aspects}.

We denote the translation of a (partial) chunk store $\Delta$ with $\fun{chr}(\Delta)$.
\end{definition}

Each chunk in a chunk store is translated to a term $\fun{chunk}/3$ that is member of a list. Note that we do not have defined the order of chunks in the list that represents the chunk store. We consider all permutations of such chunk store lists as equivalent, since for the proofs the actual choice of order will not make any difference, since we always can choose just the translation with the ``correct'' order. Where necessary, we comment on that issue in our proofs in section~\ref{sec:soundness_completeness}. The slot-value pairs in the $\fun{chunk}$ terms are sorted as defined in $\llbracket \cdot \rrbracket$.

The cognitive state $\gamma$ will be represented by $\constr{gamma}/3$ constraints that map a buffer $b \in \buffers$ to a chunk identifier and a delay. Since additional information is already represented as logical predicates, we represent them as built-in constraints in the CHR store. We can now define the translation of an abstract state.
\begin{definition}[translation of abstract states]
\label{def:translation_states}
An abstract ACT-R state $\sigma := \langle \Delta; \gamma ; \addinfo \rangle$ can be translated to the following CHR state:
\begin{align*}
\langle & \{ \constr{delta}(\fun{chr}(\Delta)) \} \\
        & \uplus \{ \constr{gamma}(b,\id_\Delta(c),d) ~|~ b \in \buffers \land \gamma(b) = (c,d) \land c = (t,\val) \} ; \addinfo ; \emptyset \rangle 
\end{align*}
We denote the translation of an ACT-R state $\sigma$ by $\fun{chr}(\sigma)$.
\end{definition}

The chunk store is represented by a $\constr{delta}$ constraint that contains the translated chunk store as defined in definition~\ref{def:translation_chunk_store}. Hence, a valid translation of an ACT-R state can only contain exactly one $\constr{delta}$ constraint.

For every buffer of the given architecture, a constraint $\constr{gamma}$ with buffer name, chunk id and delay is added to the state. Since $\gamma$ is a total function, every buffer has exactly one $\constr{gamma}$ constraint. Additionally, the chunk id in the $\constr{gamma}$ constraint must appear in exactly one $\constr{chunk}$ constraint because the co-domain of gamma refers to $\Delta$ and the chunk ids are unique.

Additional information is used directly as built-in constraints.

\subsection{Translation of Rules}
\label{sec:translation_of_rules}

In our translation scheme, ACT-R rules are translated to corresponding CHR rules. 

\subsubsection{Auxiliary Functions for Variable Names}

To manage relations between newly introduced variables, we define some auxiliary functions. The functions all produce variable names from the set of variables $\vars$ for a set of arguments that are from the set of constants $\consts$ (or a subset of it) and are applied during the translation, i.e. they do not appear in the generated CHR code.
\begin{definition}[variable functions]
Let $\vars, \consts$ be the set of variables and constants of an ACT-R architecture respectively, $\buffers \subset \consts$ the set of buffers of this architecture and $\vars_i \subset \vars$ for $i = 1, \dots, 6$ are disjoint subsets of the set of variables. Then the following auxiliary functions are defined:
\begin{description}
 \item[Chunk variable function] $\cvar : \buffers \rightarrow \vars_1, b \mapsto C^1_b$ that returns a fresh, unique variable $C^1_b$ for each buffer $b$. It identifies the chunk of a particular buffer in the translation.
 \item[Result variable functions] $\fun{resstore}, \fun{resid}, \fun{resdelay} : \buffers \rightarrow \vars_i$ for $i = 2,3,4$ are defined as $b \mapsto C^i_b$ and return a fresh, unique variable $C^i_b$ for each buffer $b$. They are needed to memorize the results of an action.
 \item[Merge variable function] $\fun{mergeid} : \buffers \times \vars_5, (b,s) \mapsto V_{b,s}$ that returns a fresh, unique variable $C^6_b$ for each buffer $b$. It is needed to memorize the new chunk identifier after merging the chunk in $b$ with the existing chunk store.
 \item[Cognitive state variable function] $\fun{cogstate} : 2^\buffers \rightarrow 2^{\buffers \times \vars_6}, \fun{cogstate}(B) \mapsto [ (b,V_b) ~:~ b \in B \land V_b \in \vars_6 ]$ (where the list is sorted by the $b$). The function returns a set of buffer-variable pairs for a set of buffers. From $\fun{cogstate}(\buffers) = \llbracket \gamma \rrbracket$ it follows that $V_b = \gamma(b)$, i.e. $\fun{cogstate}(\buffers)$ can be considered as a pattern for the cognitive state.
\end{description}
\end{definition}

\subsubsection{Built-in Constraints for Actions}

We define some built-in constraints that are needed for the translation. The idea is that they calculate the results of actions, chunk merging and chunk mapping as defined in the operational semantics of ACT-R. For actual instantiations of the abstract semantics (i.e. with a defined set of actions and chunk merging and mapping mechanisms), it has to be shown that their CHR implementations obey the properties that we define in the following.
\begin{definition}[action built-ins]
\label{def:builtin:actions}
Let $\alpha := a(b,t,P)$ be an action and $\sigma := \langle \Delta ; \gamma ; \addinfo \rangle$ an ACT-R state. Let $D := \fun{chr}(\Delta)$ be the CHR representation of the chunk store in $\sigma$ and $G := \llbracket \gamma \rrbracket$ be the enumerative list representation of the cognitive state $\gamma$. 

For all $\Delta^*, \gamma^*, \addinfo^*$: If $(\Delta^*, \gamma^*, \addinfo^*) \in I(\alpha,\sigma)$ is the result of the interpretation of $\alpha$ in $\sigma$ with $\gamma^*(b) := (c_b^*,d_b^*)$, the built-in constraint $\constr{action}/6$ is defined as follows:
\begin{align*}
\constr{action}(\alpha, D, G, D_\mathit{res}, C_\mathit{res}, E_\mathit{res})
 \land \addinfo ~\leftrightarrow~ & D_\mathit{res} = \fun{chr}(\Delta^*) \land C_\mathit{res} = \id_{\Delta^*}(c_b^*) \land E_\mathit{res} = d_b^* \\ \land & \addinfo \land \addinfo^*
\end{align*}

\end{definition}
Intuitively, the $\constr{action}$ built-in constraint represents a function that gets the action $\alpha$, CHR representations of the constraint store $\Delta$ and the cognitive state $\gamma$ as input and returns by the help of the additional information $\addinfo$ the CHR representation of $I(\alpha,\sigma)$. The constraint theory of the $\constr{action}$ constraint is well-defined, since in definition~\ref{def:interpretation_action}, $\gamma^*$ has the domain $\{ b \}$ and co-domain $\Delta^* \times \mathbb{R}_0^+$, hence $\gamma^*(b)$ is defined. Since $c_b^* \in \Delta^*$, $\id_{\Delta^*}(c_b^*)$ is defined.

\subsubsection{Built-in Constraints for Chunk Merging}

In the abstract semantics, chunk stores are merged by the operator $\chunkmerge$. In the following, built-in constraints are defined that implement $\chunkmerge$ and the corresponding mapping function $\fun{map}$.

\begin{definition}[merge built-in]
\label{def:builtin:merge}
For a set of chunk stores $\{ \Delta_1, \dots, \Delta_n \}$ with $D_i := \fun{chr}(\Delta_i)$ for $i = 1, \dots, n$, the built-in $\constr{merge}/2$ is defined as follows:
\begin{equation*}
\constr{merge}([D_1, D_2, \dots D_n], D) \leftrightarrow D = \fun{chr}(\Delta_1 \chunkmerge \Delta_2 \chunkmerge \dots \chunkmerge \Delta_n).
\end{equation*}

\end{definition}

\begin{definition}[map built-in]
\label{def:builtin:map}
For two CHR representations of chunk stores $D := \fun{chr}(\Delta)$ and $D' := \fun{chr}(\Delta')$, the built-in constraint $\constr{map}/4$ is defined as:
\begin{align*}
 \constr{map}(D,D',C,C') \leftrightarrow & C' = \id_{\Delta \chunkmerge \Delta'}(\fun{map}_{\Delta,\Delta'}(\id_{\Delta}^{-1}(C))) \text{ if $\fun{chunk}(C,T,P)$ in $D$ or} \\
  & C' = \id_{\Delta \chunkmerge \Delta'}(\fun{map}_{\Delta,\Delta'}(\id_{\Delta'}^{-1}(C))) \text{ otherwise.}
\end{align*}
The main difference to the definition of the function $\fun{map}$ from the abstract semantics is that the built-in constraint operates on chunk identifiers whereas the function operates directly on chunks. Note that in the case that the chunk with identifier $C$ appears in neither in $D$ nor in $D'$, $C'$ is bound to $\nil$ by definition of $\id$ (see definition~\ref{def:chunk_store}). 
\end{definition}

\subsubsection{List Operations}

We use the built-in constraint $\constr{in}/2$ to denote that a term is member of a list.

\begin{definition}[member of a list]
\label{def:builtin:chunk_in}
For a term $c$ and a list $l$, the constraint $c ~ \constr{in} ~ l$ holds, iff there is a term $c'$ that is member of $l$ and $c = c'$.
\end{definition}
Note that variables in $c$ are bound to the values in $l$ by this definition.

\subsubsection{Translation Scheme for Rules}

We can now define the translation scheme for rules.
\begin{definition}[translation of rules]
\label{def:translation_rules}
An ACT-R rule in set-normal form $r := L \Rightarrow R$ can be translated to a CHR rule of the following form:
\begin{align*}
r ~@~  &  \constr{delta}(D) \uplus \{ \constr{gamma}(b,\fun{cvar}(b),\fun{dvar}(b)) ~|~ b \in \buffers \} \\
       \Leftrightarrow \\
       &\bigwedge_{\test{b}{t}{P} \in L} ( \fun{chunk}(\fun{cvar}(b),t,P) ~\constr{in}~ D \land \fun{dvar}(b) {=} 0 ) ~|~ \\
        & \{ \constr{delta}(D^*)\}  \\
        \uplus~    &  \{ \constr{gamma}(b,\fun{mergeid}(b),\fun{resdelay}(b)) ~|~ a(b,t,P) \in R \} \\
        \uplus~    &  \{ \constr{gamma}(b,\fun{cvar}(b),\fun{dvar}(b)) ~|~ a(b,t,P) \notin R \} ,\\
                  & \bigwedge_{\alpha = a(b,t,P) \in R}  \constr{action}(\alpha, D, \fun{cogstate}(\buffers), \fun{resstore}(b), \fun{resid}(b), \fun{resdelay}(b)) \\
       \land~      & \constr{merge}([\fun{resstore}(b) : a(b,t,P) \in R],D') \}\\
       \land~      & \constr{merge}([D,D'],D^*) ~\uplus\\
       \land~      & \bigwedge_{a(b,t,P) \in R} \constr{map}(D,D',\fun{resid}(b),\fun{mergeid}(b)).
\end{align*}
Note that ACT-R constants and variables from $\consts$ and $\vars$ are implicitly translated to corresponding CHR variables. 

We denote the translation of a rule $r$ by $\fun{chr}(r)$ and the translation of an ACT-R model $\Sigma$ that is a set of ACT-R rules by $\fun{chr}(\rules)$. Thereby, $\fun{chr}(\rules) := \{ \fun{chr}(r) ~|~ r \in \rules \}.$
\end{definition}
The intuition behind the translation can be described as follows:

The CHR rule tests the state for a $\constr{delta}/1$ constraint representing the chunk store and $\constr{gamma}/3$ constraints that come from the buffer tests of the rule. In the $\constr{gamma}/3$ constraints a variable for the chunk identifier is introduced. In the guard, the built-in constraint $\constr{in}/2$ checks, if the chunk store represented as a list contains a term $\fun{chunk}/3$ with the same type and slot-value pairs as specified in the buffer tests. The connection to the buffer is realized by the same variable for the chunk identifier (through the variable function $\fun{cvar}$). Since the rule is in set-normal form, the buffer tests are already completed (i.e. all slots are tested) and represented as a sorted list of slot-value pairs as in the state.

In the body of the rule, the built-in constraint calculates the result of each action from the right-hand side of the ACT-R rule. The resulting chunk stores are merged to one store $D'$ by the built-in constraint $\fun{merge}$. Note that the order of merging is not specified by the translation scheme (as it is not specified by the ACT-R abstract semantics). 

The built-in constraint $\fun{map}/4$ implements the $\fun{map}$ function of the ACT-R semantics and gives access to the possibly modified chunk identifier of all elements in $D'$. By definition of $\fun{map}$, only chunk identifiers in $D'$ are modified by merging. The resulting chunk identifiers are bound to a variable specified by the variable function $\fun{mergeid}/1$. 

The resulting chunk store $\constr{delta}$ and $\constr{gamma}$ constraints for all buffers are added. If the buffers have been modified, the $\constr{gamma}$ constraints points to the resulting chunk of the action, if not it shows to the chunk that has been in the buffer before.

\subsection{No Rule Transition}

In addition to transitions by rule applications, ACT-R can also have state transitions without rule applications. This is useful for instance, if no rule is applicable (i.e. computation is stuck in a state) but there are pending requests, then simulation time can be forwarded to the point where the next request is finished and its results are visible to the procedural system. This may trigger new rules and continue the computation.

The \emph{no rule} transition can be modeled in CHR by one individual generic rule:
\begin{equation*}
no ~@~ \constr{gamma}(B,C,D) \Leftrightarrow D > 0 ~|~ \constr{gamma}(B,C,0)
\end{equation*}
This transition is possible for all requests that are pending (i.e. that have a delay $D > 0$). Hence, the system chooses one request non-deterministically.

\section{Soundness and Completeness of the Translation}
\label{sec:soundness_completeness}

In this section, we show that our translation is sound and complete w.r.t. the abstract semantics of ACT-R. This means that every transition that is possible in the abstract ACT-R semantics is also possible in CHR and vice versa.

The first step is to show that the results of the built-in constraints in the body of a translated ACT-R rule are equivalent to the interpretation of the right-hand side of the ACT-R rule. Therefore, we use that by definition the built-in constraints $\constr{action}, \constr{merge}$ and $\constr{map}$ are equivalent to their ACT-R counterparts. Additionally, we have to show that the combination of the results of individual actions leads to the same result as the built-in constraints. We use induction to show that in the following lemma.
\begin{lemma}[equivalence of effects]
\label{lemma:soundness_actions}
For an ACT-R rule $r := L \Rightarrow R$ in set-normal form, a state $\sigma := \langle \Delta,\gamma,\addinfo,t \rangle$ and $D = \fun{chr}(\Delta)$. Let $I(r,\sigma) := (\Delta^*,\gamma^*,\addinfo^*)$ with $ \bigwedge_{b \in \fun{dom}(\gamma^*)} \gamma^*(b) = (c_b^*,d_b^*)$. Then the following two propositions are equivalent:
\begin{enumerate}
 \item \begin{align*}
        \bigwedge_{\alpha_b \in R} \constr{action}(\alpha_b, D, G, \fun{resstore}(b), \fun{resid}(b), \fun{resdelay}(b)) \land \addinfo  ~\land\\
        \constr{merge}([\fun{resstore}(b) : a(b,t,P) \in R],D^*) ~\land\\
        \constr{merge}([D,D^*],D') ~\land\\
          \bigwedge_{\alpha_b \in R} \constr{map}(D,D^*,\fun{resid}(b),\fun{mergeid}(b))
       \end{align*}
 \item \begin{align*}
     D^* = \fun{chr}(\Delta^*) \land D' = \fun{chr}(\Delta \chunkmerge \Delta^*) \land \\\bigwedge_{b \in \fun{dom}(\gamma^*)} (\fun{mergeid}(b) = \id_{\Delta \chunkmerge \Delta^*}(\fun{map}_{\Delta,\Delta^*}(c_b^*)) \land \fun{resdelay}(b) = d_b^*) \land \addinfo^* \land \addinfo
       \end{align*}
\end{enumerate}
\end{lemma}
\begin{proof}
We use induction over the number of actions in $R$.

\begin{description}
 \item[base case: $|R| = 1$] Let $r := L \Rightarrow R$ be an ACT-R rule in set-normal form with one action $\alpha \in R$ and $\sigma := \langle \Delta,\gamma,\addinfo,t \rangle$. Let $D = \fun{chr}(\Delta)$ and $G = \llbracket \gamma \rrbracket$. Let $I(r,\sigma) := (\Delta^*,\gamma^*,\addinfo^*)$ with $ \bigwedge_{b \in \fun{dom}(\gamma^*)} \gamma^*(b) = (c_b,d_b)$ and $I(\alpha,\sigma) := (\Delta_\alpha,\gamma_\alpha,\addinfo_\alpha)$ with $(c_\alpha,d_\alpha) = \gamma_\alpha(b)$. 
 
 We start with
 \begin{align*}
       \constr{action}(\alpha, D, G, D_\mathit{res}, C_\mathit{res}, E_\mathit{res}) \land \addinfo  ~\land\\
        \constr{merge}([D_\mathit{res}],D^*) ~\land\\
        \constr{merge}([D,D^*],D') ~\land\\
          \constr{map}(D,D^*,C_\mathit{res},C_\mathit{res}').
       \end{align*}
       
       First of all, we reduce the $\constr{action}$ constraint by use of definition~\ref{def:builtin:actions}: 
       \begin{align*}
    \constr{action}(\alpha, D, G, D_\mathit{res}, C_\mathit{res}, E_\mathit{res})
    \land \addinfo ~\leftrightarrow~ & D_\mathit{res} = \fun{chr}(\Delta_\alpha) \land C_\mathit{res} = \id_{\Delta_\alpha}(c_\alpha) \land E_\mathit{res} = d_\alpha \\ \land & \addinfo \land \addinfo_\alpha.
    \end{align*}
   
   By definition~\ref{def:builtin:merge}, we can reduce the $\constr{merge}$ constraints to
   \begin{equation*}
    D^* = \fun{chr}(\fun{chr}^{-1}(D_\mathit{res})) \land D' = \fun{chr}(\fun{chr}^{-1}(D) \chunkmerge \fun{chr}^{-1}(D^*))
   \end{equation*}
  which is by definition of $D := \fun{chr}(\Delta)$ from the assumptions and $D_\mathit{res} = \fun{chr}(\Delta_\alpha)$ from the last step equivalent to
  \begin{equation*}
   D^* = \fun{chr}(\Delta_\alpha) \land D' = \fun{chr}(\Delta \chunkmerge \Delta_\alpha).
  \end{equation*}

  Since $C_\mathit{res} = \id_{\Delta_\alpha}(c_\alpha)$ and therefore a $\fun{chunk}$ term with identifier $C_\mathit{res}$ appears in $D^*$ (and not in $D$), we can now reduce the $\constr{map}$ built-in by definition~\ref{def:builtin:map} and get together with the definitions of $D$ and $D^*$ to
  \begin{equation*}
\constr{map}(D,D^*,C_\mathit{res},C_\mathit{res}') \leftrightarrow C_\mathit{res}' = \id_{\Delta \chunkmerge \Delta_\alpha}(\fun{map}_{\Delta,\Delta_\alpha}(\id_{\Delta_\alpha}^{-1}(C_\mathit{res}))).
\end{equation*}
Since we have that $C_\mathit{res} = \id_{\Delta_\alpha}(c_\alpha)$, this is equivalent to
     \begin{equation*}
C_\mathit{res}' = \id_{\Delta \chunkmerge \Delta_\alpha}(\fun{map}_{\Delta,\Delta_\alpha}(c_\alpha)).
\end{equation*}
By definition~\ref{def:interpretation_rule}, we have that for one action $\Delta_\alpha = \Delta^*$ and $c_\alpha = c^*$. Hence, this is equivalent to
     \begin{equation*}
C_\mathit{res}' = \id_{\Delta \chunkmerge \Delta^*}(\fun{map}_{\Delta,\Delta^*}(c^*)).
\end{equation*}

All in all, this proves the proposition for $|R| = 1$.

   
    \item[induction step: $|R| \rightarrow |R| + 1$]
    Let $\sigma := \langle \Delta ; \gamma ; \addinfo \rangle$ be an ACT-R state and $D := \fun{chr}(\Delta)$ the CHR representation of the chunk store and $G := \llbracket \gamma \rrbracket$ the enumerative list representation of the cognitive state.
    
    Let $r' := L \Rightarrow R'$ with $R' := R \cup \alpha$ be a rule that has been constructed from a rule $r := L \Rightarrow R$. Let $I(r,\sigma) := (\Delta^*,\gamma^*,\addinfo^*)$ with $ \bigwedge_{b \in \fun{dom}(\gamma^*)} \gamma^*(b) = (c_b^*,d_b^*)$ be the interpretation of the smaller rule $r$ with $|R|$ actions. 
    
    Let $I(r',\sigma) := (\Delta^{**},\gamma^{**},\addinfo^{**})$ with $\bigwedge_{b \in \fun{dom}(\gamma^{**})} \gamma^{**}(b) = (c_b^{**},d_b^{**})$ be the interpretation of the rule $r'$ that has $|R| + 1$ actions..
    
    We begin with
    \begin{align*}
        \bigwedge_{\alpha_b \in R'} \constr{action}(\alpha_b, D, G, D_b, C_b, E_b) \land \addinfo  ~\land\\
        \constr{merge}([D_b : a(b,t,P) \in R'],D^{**}) ~\land\\ 
         \constr{merge}([D,D^{**}],D'') ~\land\\
          \bigwedge_{\alpha_b \in R'} \constr{map}(D,D^{**},C_b,C_b'').
       \end{align*}    
       
       We need to apply the induction hypothesis. Therefore, we split the conjunctions and lists and get the equivalent formula
       \begin{align*}
        \bigwedge_{\alpha_b \in R} \constr{action}(\alpha_b, D, G, D_b, C_b, E_b) \land \constr{action}(\alpha_b, D, G, D_\alpha, C_\alpha, E_\alpha)  \land \addinfo  ~\land\\
        \constr{merge}([D_b : a(b,t,P) \in R] ++ [D_\alpha],D^{**}) ~\land\\ 
         \constr{merge}([D,D^{**}],D'') ~\land\\
          \bigwedge_{\alpha_b \in R} \constr{map}(D,D^{**},C_b,C_b'') \land \constr{map}(D,D^{**},C_\alpha,C_\alpha').
       \end{align*}
       
       By definition~\ref{def:builtin:actions} of $\constr{merge}$   
       and associativity of $\chunkmerge$ and neutral element $[] = \fun{chr}(\emptyset)$, we can split the merging as follows: We first merge the actions in $R$ to $D^*$ and then merge $D^*$ with the result chunk of action $\alpha$ to $D^{**}$:
       \begin{align*}
        \bigwedge_{\alpha_b \in R} \constr{action}(\alpha_b, D, G, D_b, C_b, E_b) \land \constr{action}(\alpha, C_\alpha, T_\alpha, P_\alpha, D_\alpha) \land \addinfo  ~\land\\
        \constr{merge}([D_b : a(b,t,P) \in R],D^*) ~\land \\
         \constr{merge}([D^*,D_\alpha],D^{**}) ~\land\\
         \constr{merge}([D,D^{**}],D'') ~\land\\
         \bigwedge_{\alpha_b \in R} \constr{map}(D,D^{**},C_b,C_b'') \land \constr{map}(D,D^{**},C_\alpha,C_\alpha').
       \end{align*}
       
       We can now introduce an intermediate result chunk store that merges the original store $D$ with the results from $r$, i.e. $D^*$, to a chunk store $D'$. We introduce some auxiliary variables $C_b'$ that map the chunk identifiers of the intermediate chunk store $D^*$ to the resulting chunk store $D^*$:
        \begin{align*}
        \bigwedge_{\alpha_b \in R} \constr{action}(\alpha_b, D, G, D_b, C_b, E_b) \land \constr{action}(\alpha, C_\alpha, T_\alpha, P_\alpha, D_\alpha) \land \addinfo  ~\land\\
        \constr{merge}([D_b : a(b,t,P) \in R],D^*) ~\land \\
         \constr{merge}([D,D^*],D') ~\land\\
         \bigwedge_{\alpha_b \in R} \constr{map}(D,D^*,C_b,C_b') ~\land\\
         \constr{merge}([D^*,D_\alpha],D^{**}) ~\land\\
         \constr{merge}([D,D^{**}],D'') ~\land\\
         \bigwedge_{\alpha_b \in R} \constr{map}(D,D^{**},C_b,C_b'') \land \constr{map}(D,D^{**},C_\alpha,C_\alpha').
       \end{align*}
             
       We can now apply the induction hypothesis:
       \begin{align*}
       (*) := &\\
         &\addinfo \land \addinfo^* ~\land\\
         &D^* = \fun{chr}(\Delta^*) \land D' = \fun{chr}(\Delta \chunkmerge \Delta^*) ~\land\\
         &\bigwedge_{b \in \fun{dom}(\gamma^*)} (C_b' = \id_{\Delta \chunkmerge \Delta^*}(\fun{map}_{\Delta,\Delta^*}(c_b^*)) \land E_b = d_b^*)  ~\land\\
         &\constr{merge}([D^*,[\fun{chunk}(C_\alpha,T_\alpha,P_\alpha)]],D^{**}) ~\land\\
         &\constr{merge}([D,D^{**}],D'') ~\land\\
         &\bigwedge_{\alpha_b \in R} \constr{map}(D,D^{**},C_b,C_b'') \land \constr{map}(D,D^{**},C_\alpha,C_\alpha').
       \end{align*}
       Thereby, it holds by definition of $I(r,\sigma)$ that $\fun{dom}(\gamma^*) = \{b ~|~ \alpha_b \in R \}$. 

       If we apply definition~\ref{def:builtin:actions} to the remaining $\constr{action}$ constraint of action $\alpha$, we get
       \begin{align*}
       \constr{action}(\alpha, D, G, D_\alpha, C_\alpha, E_\alpha)
 \land \addinfo ~\leftrightarrow~ & D_\alpha = \fun{chr}(\Delta_\alpha) \land C_\alpha = \id_{\Delta_\alpha}(c_\alpha) \land E_\alpha = d_\alpha \\ \land & \addinfo \land \addinfo^*.
      \end{align*}
     
       Hence, we have that 
       \begin{equation*}
        \constr{merge}([D^*,D_\alpha,D^{**}) \leftrightarrow D^{**} = \fun{chr}(\Delta^* \chunkmerge \Delta_\alpha).
       \end{equation*}
        Thus, $D^{**}$ is the CHR version of the merging of the results from the actions in $R$ merged with the results from $\alpha$. By definition~\ref{def:interpretation_rule} of the interpretation of rules, we have that
        \begin{equation*}
         I(r',\sigma) = (\Delta^* \chunkmerge \Delta_\alpha,\gamma' \cup \gamma_\alpha, \addinfo^* \land \addinfo_\alpha).
        \end{equation*}
        This is equivalent to
        \begin{equation*}
         I(r',\sigma) = (\Delta^{**},\gamma^{**}, \addinfo^{**})
        \end{equation*}
        by definition of $r, r'$ and the interpretation of rules (definition~\ref{def:interpretation_rule}).
        
        We now reduce the last $\constr{merge}$ constraint to:
        \begin{equation*}
         \constr{merge}([D,D^{**}],D'') \leftrightarrow D'' = \fun{chr}(\Delta \chunkmerge \Delta^{**})
        \end{equation*}
	
	We can now reconnect the $\constr{map}$ constraints and get
	\begin{align*}
	 \bigwedge_{\alpha_b \in R} \constr{map}(D,D^{**},C_b,C_b'') \land \constr{map}(D,D^{**},C_\alpha,C_\alpha') \leftrightarrow  & \bigwedge_{\alpha_b \in R'} \constr{map}(D,D^{**},C_b,C_b'') 	 
	\end{align*}
	
	By definition~\ref{def:unified_operational_semantics}, for all defined buffers $\gamma^{**}(b) = \fun{map}_{\Delta,\Delta^{**}}(c_b)$ and we have assumed that $\gamma^{**}(b) = c_b^{**}$. By definition~\ref{def:builtin:map}, this yields
	\begin{equation*}
	  C_b'' = \id_{\Delta \chunkmerge \Delta^{**}}(\fun{map}_{\Delta,\Delta^{**}}(\id_{\Delta^{**}}^{-1}(C_b)))
         \end{equation*}.
         $\id_{\Delta^{**}}^{-1}(C_b)$ exists, since the co-domain of $\gamma^{**}$ is $\Delta^{**} \times \mathbb{R}^+_0$ by definition~\ref{def:partial_chunk_store} and the $C_b$ are identifiers for the chunks in $\gamma^{**}$.
        
        If we apply all this to the conjunction in $(*)$, we get 
         \begin{align*}
     D^{**} = \fun{chr}(\Delta^{**}) \land D'' = \fun{chr}(\Delta \chunkmerge \Delta^{**}) ~\land\\
     \bigwedge_{b \in \fun{dom}(\gamma^{**})} (C_b'' = \id_{\Delta \chunkmerge \Delta^{**}}(\fun{map}_{\Delta,\Delta^{**}}(c_b^{**})) \land E_b = d_b^{**}) \land \addinfo^{**} \land \addinfo
       \end{align*}
\end{description}
\end{proof}

The next lemma proves that rule application transitions are sound.
\begin{lemma}[soundness of rule applications]
\label{lemma:soundness_rules}
For all ACT-R rules $r$ and ACT-R states $\sigma, \sigma' \in \states_\abstr:$ if $\sigma \utrans^r \sigma'$ then $\fun{chr}(\sigma) \mapsto^r \fun{chr}(\sigma')$. 
\end{lemma}
\begin{proof}
Let $r := L \Rightarrow R$ and $\sigma := \langle \Delta ; \gamma ; \addinfo ; 0 \rangle$. Since $r \matches \sigma$, we know that for every buffer test $\beta := \test{b}{t}{P}$ there is a substitution $\theta$ such that $\gamma(b) = ((t,\val),0)$ and for all $(s,v) \in P : \id_\Delta(\val(s)) = v\theta$. 

Let $\rho := \fun{chr}(\sigma)$ the CHR translation of the ACT-R state $\sigma$. By definition~\ref{def:translation_states}, we have that 
\begin{align*}
\rho \equiv \langle & \{ \constr{delta}([\fun{chunk}(\id_\Delta(c),t,\llbracket \val \rrbracket) ~:~ c \in \Delta \land c = (t,\val)]) \} \\
                    & \uplus \{ \constr{gamma}(b,\id_\Delta(c),d) ~|~ b \in \buffers \land \gamma(b) = (c,d) \} ; \addinfo ; \emptyset \rangle
\end{align*}

In the next step we split the state to only concentrate on the parts we are interested in for the rule application, i.e. $\constr{gamma}$ constraints for buffers that occur in a test. We do the same by splitting the representation of the constraint store to the chunks that occur in buffers and all other chunks. This is possible, since $\fun{chr}(\Delta)$ does not define an order on the $\fun{chunk}$ terms in the list in the $\constr{delta}$ constraint. 
\begin{align*}
\rho \equiv \langle & \{ \constr{delta}([\fun{chunk}(\id_\Delta(c),t,\llbracket \val \rrbracket) ~:~ \test{b}{t'}{P} \in L \land \gamma(b) = (c,0) \land c = (t,\val)] {++} \mathbb{D}) \} \\
                    & \uplus \{ \constr{gamma}(b,\id_\Delta(c),d) ~|~ \test{b}{t}{P} \in L \land \gamma(b) = (c,d) \} \uplus \mathbb{G} ; \addinfo  ; \emptyset \rangle
\end{align*}
Note that by now we only have rewritten the CHR state without requiring that the tests we refer to match. 

Due to the fact that $r \matches \sigma$, we can apply this knowledge to the translated state:
\begin{align*}
\rho \equiv \langle & \{ \constr{delta}([\fun{chunk}(\id_\Delta(c),t,\llbracket \val \rrbracket) ~:~ \test{b}{t}{P} \in L \land \gamma(b) = (c,0) \land c = (t,\val)\\
                    & ~~\land \text{ for all }(s,v) \in P : \id_\Delta(\val(s)) = v\theta] {++} \mathbb{D}) \} \\
                    & \uplus \{ \constr{gamma}(b,\id_\Delta(c),0) ~|~ \test{b}{t}{P} \in L \land \gamma(b) = (c,0) \} \uplus \mathbb{G} ; \addinfo  ; \emptyset \rangle
\end{align*}
Since $r$ is in set-normal form, we can assume that there is only one test for each buffer in $L$. Hence, we find a $\constr{gamma}$ constraint in $\rho$ for each such test without violating definition~\ref{def:translation_states}, that only produces one $\constr{gamma}$ constraint for each buffer.

Due to set-normal form of $r$ and hence totality of $P$ w.r.t. the domain of $\val$, it is clear that $P\theta = \llbracket \val \rrbracket$. Hence, we can reduce the state to: 
\begin{align*}
\rho \equiv \langle & \{ \constr{delta}([\fun{chunk}(\id_\Delta(c),t,P\theta) ~:~ \test{b}{t}{P} \in L \land \gamma(b) = (c,0)] {++} \mathbb{D}) \} \\
                    & \uplus \{ \constr{gamma}(b,\id_\Delta(c),0) ~|~ \test{b}{t}{P} \in L \land \gamma(b) = (c,0) \} \uplus \mathbb{G} ; \addinfo  ; \emptyset \rangle
\end{align*}

We introduce a substitution with fresh variables to replace the chunk identifiers in the state, i.e. $\theta' := \{ \cvar(b)/\id_\Delta(c) ~|~ \gamma(b) = c \}$. 
\begin{align*}
\rho \equiv \langle & \{ \constr{delta}([\fun{chunk}(\fun{cvar}(b)\theta',t,P\theta) ~:~ \test{b}{t}{P} \in L] {++} \mathbb{D}) \} \\
                    & \uplus \{ \constr{gamma}(b,\fun{cvar}(b)\theta',0) ~|~ \test{b}{t}{P} \in L \} \uplus \mathbb{G} ; \addinfo  ; \emptyset \rangle
\end{align*}

Let $\Theta$ be the conjunction of syntactic equality constraints that can be derived from the substitution $\theta \cup \theta'$, i.e. each substitution $x/t \in \theta \cup \theta'$ appears in $\Theta$ as $x=t$. By definition~\ref{def:chr_state_equiv:substitution}, we can move the substitution to the built-in store:
\begin{align*}
\rho \equiv \langle & \{ \constr{delta}([\fun{chunk}(\fun{cvar}(b),t,P) ~:~ \test{b}{t}{P} \in L] {++} \mathbb{D}) \} \\
                    & \uplus \{ \constr{gamma}(b,\fun{cvar}(b),0) ~|~ \test{b}{t}{P} \in L \} \uplus \mathbb{G} ; \Theta \uplus \addinfo  ; \emptyset \rangle
\end{align*}

We introduce a fresh variable $D$ and add $D{=}[\fun{chunk}(\fun{cvar}(b),t,P) ~:~ \test{b}{t}{P} \in L] {++} \mathbb{D}$ to the built-in store, which leads to the equivalent state:
\begin{align*}
\rho \equiv \langle & \{ \constr{delta}(D) \} \uplus \{ \constr{gamma}(b,\fun{cvar}(b),0) ~|~ \test{b}{t}{P} \in L \} \uplus \mathbb{G} ;\\ & \quad \Theta \uplus D{=}[\fun{chunk}(\fun{cvar}(b),t,P) ~:~ \test{b}{t}{P} \in L] {++} \mathbb{D} \uplus \addinfo  ; \emptyset \rangle
\end{align*}

By definition~\ref{def:builtin:chunk_in}, we have that $D{=}d \leftrightarrow \{ \fun{chunk}(\fun{cvar}(b),t,P) ~\constr{in}~ D ~|~ \test{b}{t}{P} \}$. We can replace the corresponding built-ins by definition~\ref{def:chr_state_equiv:equiv}:
\begin{align*}
\rho \equiv \langle & \{ \constr{delta}(D) \} \uplus \{ \constr{gamma}(b,\fun{cvar}(b),0) ~|~ \test{b}{t}{P} \in L \} \uplus \mathbb{G} ; \\
                    & \Theta \uplus \{ \fun{chunk}(\fun{cvar}(b),t,P) ~\constr{in}~ D ~|~ \test{b}{t}{P} \} \uplus \addinfo  ; \emptyset \rangle
\end{align*}

Let $\fun{chr}(r) := H \Leftrightarrow G ~|~ B_c, B_b$. By definition~\ref{def:translation_rules}, we have that
\begin{align*}
\rho \equiv \langle & H \uplus \mathbb{G} ; \Theta \uplus G \uplus \addinfo  ; \emptyset \rangle 
\end{align*}
By definition~\ref{def:chr_operational_semantics}, $\fun{chr}(r)$ is applicable in $\rho \equiv \fun{chr}(\sigma)$ and $\rho \mapsto^r \rho'$ with
\begin{align*}
\rho' \equiv \langle & B_c \uplus \mathbb{G} ; \Theta \land G \land B_b \land \addinfo  ; \emptyset \rangle 
\end{align*}
Due to the definition~\ref{def:translation_rules} of $\fun{chr}(r)$, we have that
\begin{align*}
 B_b := & \bigwedge_{\alpha = a(b,t,P) \in R } \left(\constr{action}\left(\alpha, \fun{resid}\left(b\right), \fun{restype}\left(b\right), \fun{resslots}\left(b\right), \fun{resdelay}\left(b\right)\right)\right) ~\land \\
        & \constr{merge}([[\fun{chunk}(\fun{resid}(b),\fun{restype}(b),\fun{resslots}(b))] : a(b,t,P) \in R],D') ~\land\\
        & \constr{merge}([D,D'],D^*) ~\land\\
        & \bigwedge_{a(b,t,P) \in R} (\constr{map}(D,D',\fun{resid}(b),\fun{mergeid}(b))) \\
 B_c := & \{ \constr{delta}(D^*)\} ~\uplus \\
        &  \{ \constr{gamma}(b,\fun{mergeid}(b),\fun{resdelay}(b)) ~|~ a(b,t,P) \in R \}     
\end{align*}
Since $r$ is in set-normal form, all buffers appearing in $L$ also appear in $R$. Hence, all $\constr{gamma}$ constraints removed by $\fun{chr}(r)$ are added in $B_c$. There is exactly one $\constr{delta}$ constraint in $\rho'$. It remains to show that the $\constr{delta}$ and $\constr{gamma}$ constraints are the ones describing the state $\sigma'$.

By lemma~\ref{lemma:soundness_actions}, we have that for all $I(r,\sigma) \ni (\Delta^*,\gamma^*,\addinfo^*)$ with $\gamma^*(b) = (c_b^*,d_b^*)$ for all $b \in \fun{dom}(\gamma^*)$
\begin{align*}
 B_b \leftrightarrow &  ~D^* = \fun{chr}(\Delta^*) \land D' = \fun{chr}(\Delta \chunkmerge \Delta^*) ~\land \\ 
                     &\bigwedge_{b \in \fun{dom}(\gamma^*)} (\fun{mergeid}(b) = \id_{\Delta \chunkmerge \Delta^*}(\fun{map}_{\Delta,\Delta^*}(c_b^*)) \land \fun{resdelay}(b) = d_b^*) \land \addinfo^* \land \addinfo
\end{align*}
 All in all, we get by definition~\ref{def:translation_states} that $\rho' \equiv \fun{chr}(\sigma')$. Hence, the translation of rule applications is sound w.r.t. the abstract operational semantics of ACT-R. 
\end{proof}


We have now shown that our translation is sound regarding rule applications, i.e. a rule that can be applied in the abstract semantics of ACT-R can also be applied in the translated CHR program and leads to the same result. The state transition system of ACT-R has a second type of transitions: the \emph{no-rule transitions}, that can be applied if there is a buffer with a delay $> 0$, i.e. an invisible chunk. In that case, the no-rule transition allows us to make this chunk visible to the production system by setting the delay to zero. We now show that our translation is sound and complete regarding the no-rule transition.
\begin{lemma}[soundness and completeness of the no-rule transition]
\label{lemma:soundness_completeness_norule}
For an ACT-R model $M$, ACT-R states $\sigma$ and $\sigma'$ and their CHR counterparts $\fun{chr}(M), \fun{chr}(\sigma)$ and $\fun{chr}(\sigma')$ the following propositions are equivalent:
\begin{enumerate}
 \item $\sigma \utrans_\mathbf{no} \sigma'$ in the model $M$ \label{lemma:soundness_completeness_norule:actr}
 \item $\fun{chr}(\sigma) \mapsto^{\mathit{no}} \fun{chr}(\sigma')$ \label{lemma:soundness_completeness_norule:chr}
\end{enumerate}
\end{lemma}
\begin{proof}
We start with proposition~\ref{lemma:soundness_completeness_norule:actr}. Let $\sigma := \langle \Delta ; \gamma ; \addinfo \rangle$ and $\sigma' := \langle \Delta ; \gamma'; \addinfo \rangle$. In the abstract semantics, the no-rule transition is defined as
  \begin{equation*}
  \frac{ 
    \gamma(b^*) = (c^*,d^*) \land d^* > 0
  }{ 
    \langle \Delta ; \gamma ; \addinfo \rangle \utrans_\mathbf{no} \langle \Delta ; \gamma'; \addinfo \rangle
  }
 \end{equation*}
where $\gamma'(b) = \begin{cases}
                     (c^*,0)   & \text{if $b = b^*,$}\\
                     \gamma(b) & \text{otherwise.}
                    \end{cases}
$

The \emph{no-rule} transition in CHR is represented by the following CHR rule in $\fun{chr}(M)$:
\begin{equation*}
no ~@~ \constr{gamma}(B,C,D) \Leftrightarrow D > 0 ~|~ \constr{gamma}(B,C,0)
\end{equation*}

Since the \emph{no-rule} transition is applicable in $\sigma$ for some arbitrary but fixed $b^*$, it holds that 
\begin{equation*}
 \gamma(b^*) = (c^*,d^*) \land d^* > 0.
\end{equation*}
Therefore, in $\fun{chr}(\sigma)$ there must have the following form by definition~\ref{def:translation_states}:
\begin{equation*}
 \fun{chr}(\sigma) \equiv \langle \constr{gamma}(B,C,D) \uplus  \mathbb{G} ; B{=}b^* \land C{=}\id_\Delta(c^*) \land D{=}d^* \land \mathbb{C} ; \emptyset \rangle.
\end{equation*}
Since $d^* > 0$, this is equivalent to
\begin{equation*}
 \fun{chr}(\sigma) \equiv \langle \constr{gamma}(B,C,D) \uplus  \mathbb{G} ; D>0 \land B{=}b^* \land C{=}\id_\Delta(c^*) \land D{=}d^* \land \mathbb{C} ; \emptyset \rangle
\end{equation*}
and therefore the rule \emph{no} can be applied to $\fun{chr}(\sigma)$ by definition~\ref{def:chr_operational_semantics}. This leads to the state $\rho'$ with
\begin{equation*}
 \rho' \equiv \langle \constr{gamma}(b^*,c^*,0) \uplus  \mathbb{G} ; \mathbb{C} ; \emptyset \rangle.
\end{equation*}
By definition~\ref{def:translation_states}, we have that $\rho' \equiv \fun{chr}(\sigma)$.

The other direction is analogous.
\end{proof}


\begin{lemma}[completeness of rule applications]
\label{lemma:completeness_rules}
Let $\rho, \rho'$ be CHR states that have been translated from ACT-R states, i.e. there are ACT-R states $\sigma, \sigma'$ such that $\rho := \fun{chr}(\sigma)$ and $\rho' := \fun{chr}(\sigma')$. Let $r$ be a CHR rule translated from an ACT-R rule $s$. If $\rho \mapsto_r \rho'$ then $\sigma \utrans_{s} \sigma'$.
\end{lemma}
\begin{proof}
 The CHR rule $r := H \Leftrightarrow G ~|~ B_c, B_b$ has been translated from an ACT-R rule $s$. Let $s := L \Rightarrow R$ be an ACT-R rule in set normal form. Then $r$ has the following form by definition~\ref{def:translation_rules}:
 \begin{align*}
    &\constr{delta}(D) \uplus \{ \constr{gamma}(b,\fun{cvar}(b),\fun{dvar}(b)) ~|~ b \in \buffers \} \\
        \Leftrightarrow & \\
       & \bigwedge_{\test{b}{t}{P} \in L} \fun{chunk}(\fun{cvar}(b),t,P) ~\constr{in}~ D \land \fun{dvar}(b) {=} 0 ~|~ \\
        &     \{ \constr{delta}(D^*)\} ~\uplus \\
         &                   \{ \constr{gamma}(b,\fun{mergeid}(b),\fun{resdelay}(b)) ~|~ a(b,t,P) \in R \} ~\uplus  \\
         &                   \{ \constr{gamma}(b,\fun{cvar}(b),\fun{dvar}(b)) ~|~ a(b,t,P) \notin R \} ,\\
        & \bigwedge_{\alpha = a(b,t,P) \in R} (\constr{action}(\alpha, D, \fun{cogstate}(\buffers), \fun{resstore}(b), \fun{resid}(b), \fun{resdelay}(b)))  \\
        \land~ & \constr{merge}([\fun{resstore}(b) : a(b,t,P) \in R],D') \\
        \land~ & \constr{merge}([D,D'],D^*) \\
        \land~ & \bigwedge_{a(b,t,P) \in R} ( \constr{map}(D,D',\fun{resid}(b),\fun{mergeid}(b)))  
 \end{align*}
 where $B_b$ are the built-in constraints of the body and $B_c$ the CHR constraints.
 
Since $r$ is applicable in $\rho$, by definition~\ref{def:chr_operational_semantics}, $\rho$ must have the following form:
\begin{equation*}
 \rho \equiv  \langle H \uplus \mathbb{G} ; G \land \mathbb{C} ; \mathbb{V} \rangle
\end{equation*}
Thereby, $\mathbb{G}$ is a multi-set of user-defined constraints and $\mathbb{C}$ a conjunction of built-in constraints. Since by definition of ACT-R states, $H \uplus \mathbb{G}$ must be ground and therefore there must be some built-in constraints $\Theta$ in $\mathbb{C}$ that bind those variables to the values from the state:
\begin{equation*}
 \rho \equiv  \langle H \uplus \mathbb{G} ; \Theta \land G \land \mathbb{C'} ; \mathbb{V} \rangle
\end{equation*}
$\Theta$ is a conjunction of constraints of the form $V = c$ for a variable $V$ and a term $c$, that binds the variables from $H$ and $G$ to the values from the state, i.e. the variables $\fun{cvar}(b), \fun{dvar}(b)$ (for all $b \in \buffers$) and the variables appearing in the $P$ from the guard are bound to some values from the state. We denote $\theta$ as the substitution that follows from $\Theta$. Furthermore, $\rho$ must contain the additional information $\addinfo$:
\begin{equation*}
 \rho \equiv  \langle H \uplus \mathbb{G} ; \Theta \land G \land \addinfo \land \mathbb{C''} ; \mathbb{V} \rangle
\end{equation*}

 We now want to construct the ACT-R state $\sigma$. Let $\sigma := \langle \Delta ; \gamma ; \addinfo' \rangle$. Thereby, $\fun{chr}(\Delta) = D$ and for all $b \in \buffers : \gamma(b) = (\id_\Delta^{-1}(\fun{cvar}(b)),\fun{dvar}(b))\theta$ by definition~\ref{def:translation_states}. Additionally, $\addinfo' := \addinfo$. We can now apply the same equivalences as in lemma~\ref{lemma:soundness_rules} in reverse order and get 
\begin{align*}
\rho \equiv \langle & \{ \constr{delta}([\fun{chunk}(\id_\Delta(c),t,\llbracket \val \rrbracket) ~:~ \test{b}{t}{P} \in L \land \gamma(b) = (c,0) \land c = (t,\val)\\
                    & ~~\land \text{ for all }(s,v) \in P : \id_\Delta(\val(s)) = v\theta] {++} \mathbb{D}) \} \\
                    & \uplus \{ \constr{gamma}(b,\id_\Delta(c),0) ~|~ \test{b}{t}{P} \in L \land \gamma(b) = (c,0) \} \uplus \mathbb{G} ; \Theta \land G \land \addinfo \land \mathbb{C}'' ; \emptyset \rangle
\end{align*}
From there it is clear that $s \matches \sigma$ by definition~\ref{def:matching}. We now apply $s$ to $\sigma$ according to definition~\ref{def:unified_operational_semantics} 
and get
\begin{equation*}
 \sigma' := \langle \Delta \chunkmerge \Delta^* ; \gamma' ; \addinfo \land \addinfo^* \rangle
\end{equation*}
where  $(\Delta^*, \gamma^*, \addinfo^*) \in I(s\theta,\sigma)$ and $\gamma'(b) := \begin{cases}
                       (\fun{map}_{\Delta,\Delta^*}(c),d) & \text{if } \gamma^*(b) = (c,d) \text{ is defined}\\
                       (\fun{map}_{\Delta,\Delta^*}(c),d)          & \text{otherwise, if $\gamma(b) = (c,d)$.}
                      \end{cases}
$

By definition~\ref{def:chr_operational_semantics}, the application of $r$ in $\rho$ leads to the state $\rho'$ with
\begin{equation*}
 \rho' \equiv \langle B_c \uplus \mathbb{G} ; G \land B_b \land \Theta \land \addinfo ; \mathbb{V} \rangle.
 \end{equation*}
By lemma~\ref{lemma:soundness_actions}, we get that $B_b$ is equivalent to
\begin{align*}
     D^* = \fun{chr}(\Delta^*) \land D' = \fun{chr}(\Delta \chunkmerge \Delta^*) \land \\\bigwedge_{b \in \fun{dom}(\gamma^*)} (\fun{mergeid}(b) = \id_{\Delta \chunkmerge \Delta^*}(\fun{map}_{\Delta,\Delta^*}(c^*)) \land \fun{resdelay}(b) = d_b^*) \land \addinfo^* \land \addinfo
       \end{align*}
       
Hence, $\rho'$ and $\sigma'$ correspond directly to each other, i.e. $\rho' \equiv \fun{chr}(\sigma')$ and therefore rule transitions in the translation are complete w.r.t. the abstract operational semantics of ACT-R.
\end{proof}


\begin{theorem}[soundness and completeness of translation]
For an ACT-R model $M$, ACT-R states $\sigma$ and $\sigma'$ and their CHR counterparts $\fun{chr}(M), \fun{chr}(\sigma)$ and $\fun{chr}(\sigma')$ the following propositions are equivalent:
\begin{enumerate}
 \item $\sigma \utrans \sigma'$ in the model $M$ \label{theorem:soundness_completeness:actr}
 \item $\fun{chr}(\sigma) \mapsto \fun{chr}(\sigma')$ \label{theorem:soundness_completeness:chr}
\end{enumerate}
\end{theorem}
\begin{proof}
This follows directly from lemmas~\ref{lemma:soundness_rules}, \ref{lemma:soundness_completeness_norule} and \ref{lemma:completeness_rules}.
\end{proof}

\section{Related Work}
\label{sec:related_work}

We want to highlight the contribution of \cite{albrecht_2014a} that we have used as a starting point to improve our work by unifying and extending it by our needs. We discuss this line of work in detail in section~\ref{sec:related_work:albrecht}. In section~\ref{sec:related_work:other} we summarize other work related to this paper.

\subsection{Formal Semantics According to Albrecht et al.}
\label{sec:related_work:albrecht}

The formalization according to \cite{albrecht_2014a} has been developed independently from our work in \cite{gall_rule_based_2013,gall_lopstr2014}. Our very abstract semantics is based on it. \cite{albrecht_2014a} 
basically defines a general production rule system that works on sets of buffers and chunks without specifying actual matching, actions and effects for the sake of modularity and reusability. We briefly summarize the differences between the semantics in \cite{albrecht_2014a}  and our very abstract semantics. For details, we refer to the original papers. The nomenclature in this paper differs in some points from the original paper \cite{albrecht_2014a} to unify it with our previous work. We omit module queries for the sake of brevity.

The sets of buffers $\buffers$ and action symbols $A$ are defined as in section~\ref{sec:preliminaries:actr}. For the sake of brevity, we have omitted the so-called buffer queries in our definition of the very abstract semantics. Queries are an additional type of test on the left-hand side of a rule. The very abstract semantics can be easily extended by queries. We have adopted the definition of chunk types, chunks and cognitive states from the formalization of \cite{albrecht_2014a}, although the set of chunks in \cite{albrecht_2014a} should be a multi-set as example~\ref{ex:chunk_store_nat} shows. However, we have reduced the definition in our very abstract semantics by omitting the notion of a \emph{finite trace}, which is a sequence $\gamma_0, \gamma_1, \dots \in \cogstates^*$ of cognitive states. Those traces are used to compute the effects of an action. This definition seems inaccurate as the information of a finite trace that only logs the contents of the buffers at each step does not suffice to calculate sub-symbolic information. In typical definitions the calculation of production rule utilities needs the times of all rule applications that are not part of the trace. In other implementations and instantiations of ACT-R, there can be more additional information that is needed for sub-symbolic calculations. That is why we have extended the states by a parameter valuation function that abstracts from the information needed and leaves it to the architecture to define which information is stored.

In \cite{albrecht_2014a}, effects of actions with action symbol $\alpha \in A$ are defined by an interpretation function $I_\alpha : \traces \rightarrow 2^{\cogstatespart \times 2^\Delta}$ (we have omitted queries as stated before). Similarly to the very abstract semantics, it assigns to each finite trace the possible effects of an action. Effects are a partial cognitive state that overwrites the contents of the buffers as in the very abstract semantics and a set $C \subseteq \Delta$ that defines the chunks that are removed. In typical implementations of ACT-R, the chunks in $C$ are moved to the declarative module which explains the need to define such a set. We have generalized this information by the notion parameter valuations that can be manipulated by an interpretation function. This enables us to abstract from the specific concept of moving chunks to declarative memory in our abstract semantics for example. Note that in \cite{albrecht_2014a}, the combination of interpretation functions to a rule interpretation is only stated informally. Additionally, we have extended the domain of an action interpretation function to actions, i.e. terms over the actions symbols in $A$, and states instead of only action symbols, since more information is needed to calculate, like the parameters of the actions (i.e. the slot-value pairs) and information from the state.

The production rule selection function $S : \traces \rightarrow 2^\rules$ maps a set of applicable rules to each finite trace. In the very abstract semantics we have extended the domain from traces to a whole state since again additional information might be needed to resolve rule conflicts. With parameter valuations, we abstract from the information that is actually needed and leave it to the architecture definition. Additionally, our definition of selection function adds the notion of variable bindings that are not considered in \cite{albrecht_2014a}.

The operational semantics in \cite{albrecht_2014a} is defined as a labeled, timed transition system with the following transition relation $\rightsquigarrow$ over time-stamped cognitive states from $\cogstates \times \mathbb{R}_0^+$:
\begin{equation*}
(\gamma,t)_\pi \mathrel{\mathop{\rightsquigarrow}^{r,d,\omega}} (\gamma',t')
\end{equation*}
for a production rule $r \in \rules$, an execution delay $d \in \mathbb{R}^+_0$, a set of chunks $\omega \subseteq \Delta$ and a finite trace $\pi \in \traces$, if and only if $r \in S(\pi,\gamma)$, i.e. $r$ is applicable in $\gamma$, the actions of $r$ according to the interpretation functions yield $\gamma'$ and $t' = t + d$. 

Note that the set of chunks $\omega$ has been used but never defined in the original paper \cite{albrecht_2014a}. We suspect that it represents an equivalent to the chunk store from our abstract semantics, i.e. the used subset of all possible chunks (which is how $\Delta$ is defined according to the paper). 
Although we consider it an integral part of ACT-R, the matching of rules -- and particularly binding of variables by the matching -- is completely hidden in $S$ or even not defined. On the one hand this simplifies exchanging the matching, on the other hand the function $S$ should then be defined slightly different to enable proper handling of variable bindings and conflict resolution as we discuss in section~\ref{sec:very_abstract_semantics}. 

In the original semantics according to \cite{albrecht_2014a} 
there is no definition of what happens if there is no rule applicable, but there are still effects of e.g. requests that can be applied. We have treated this case by adding the \emph{no rule} transition to the very abstract semantics.

\subsection{Other Work}
\label{sec:related_work:other}

The reference implementation of ACT-R is described in a technical document \cite{actr_reference} that defines the operational semantics mostly verbally and determines various technical details that are important for this exact implementation but not the architecture itself. In \cite{gall_ppdp2015}, we have defined a semantics that describes the core of the reference implementation of ACT-R and show that every transition possible in this refined semantics is also possible in the abstract semantics. This shows that formal reasoning about our abstract semantics is meaningful to actual implementations.

There are approaches of implementing ACT-R in other languages, for example a Python implementation \cite{stewart_deconstructing_2007} or (at least) two Java implementations \cite{jactr,java_actr}. All those approaches do not concentrate on formalization and analysis, but only introduce new implementations. In \cite{stewart_deconstructing_2007} it is stated that exchanging integral parts of the ACT-R reference implementation is difficult due to the need of an extensive knowledge of technical details. They propose an architecture that is more concise and reduced to the fundamental concepts (that they also identify in their paper). However, their work still lacks a formalization of the operational semantics.

In \cite{albrecht_2014b}, the authors summarize the work on semantics in the ACT-R context. They also come to the conclusion that there are only new implementations available that sometimes try to formalize parts of the architecture, but no formal definition of ACT-R's operational semantics. The authors use this result as a motivation for their work in \cite{albrecht_2014a}.

We describe an adaptable implementation of ACT-R using Constraint Handling Rules (CHR) in \cite{gall_rule_based_2013,gall_iclp_2014,gall_iccm_2015} that is based on our formalization. Due to the declarativity of CHR, the implementation is very close to the formalization and easy to extend. This has been proved by exchanging the conflict resolution mechanism (that is an integral part of typical implementations) with very low effort \cite{gall_iclp_2014}. Even the integration of refraction, i.e. inhibiting rules to fire twice on the same (partial) state, has been exemplified and can be combined with other conflict resolution strategies. The translation presented there is close to both the core of the reference implementation and the abstract semantics whose abstract parts are defined such that the match the reference implementation and some of its extensions.

\section{Conclusion}

In this paper, we have defined a very abstract operational semantics for ACT-R that can serve as a common base to analyze other operational semantics since it leaves enough room for various ACT-R variants. We then have refined this semantics to an abstract semantics. 

Similar to the very abstract semantics, the abstract semantics abstracts from details like timings, latencies, forgetting, learning and specific conflict resolution. However, it defines the matching of rules and the processing of actions as they are typically found in ACT-R implementations. Hence, the abstract semantics concentrates on an abstract version of the typical implementations of ACT-R's procedural system. This makes it possible to reason about the general transitions that are possible in many ACT-R implementations.

We have defined a translation of ACT-R models to Constraint Handling Rules (CHR) that is sound and complete w.r.t. our abstract semantics of ACT-R. 
To the best of our knowledge, the abstract semantics together with the sound and complete translation to CHR is the first formal formulation of ACT-R that is suitable to implementation.

For the future, we want to investigate how we can use our abstract semantics in practice, since the faithful embedding in CHR opens many possibilities to reason about cognitive models by applying theoretical results from the CHR world to ACT-R models. For instance, confluence is the property that a program always yields the same result for the same input regardless of the order rules are applied. In CHR, there is a decidable confluence criterion for terminating programs \cite{fru_chr_book_2009}. Although the human mind is probably not confluent because there are many competing strategies with different outputs for the same task, there are always sequences of rules in cognitive models that should not be interfered by any other rule. A confluence criterion helps identifying the parts of the model that are not confluent. This can improve model quality by allowing for controlled non-confluence where desired guaranteeing the rules in the rest of the program not interfering with each other.

However, in practice confluence usually is too strict. With the notion of invariant-based confluence \cite{duck_stuck_sulz_observable_confluence_iclp07} only valid states that can be reached are considered, making confluence analysis applicable for practical use. Confluence modulo equivalence \cite{christiansen_confmeq_arxiv_2016} is a recent approach to test if programs always yield states that are considered equivalent by an arbitrary definition of state equivalence for the same input. This could be used to analyze if an ACT-R model always yields a certain class of chunks for the same input. For instance it could be interesting for the modeler to know if a certain buffer always contains a chunk of a certain chunk type or with a certain value in some slot at the end of a computation. By that method, models could guarantee certain properties on their final states improving explanatory power and quality of cognitive models.

To make predictions on the probability that a cognitive model has a certain result, we plan to use the CHR extension CHRiSM \cite{sneyers2010chr,sneyers2009chrism} that allows to enrich CHR rules with probabilities. It supports probability computation and even an expectation-maximization learning algorithm that could be used for parameter learning of cognitive models.

%
%

\bibliographystyle{ACM-Reference-Format}
\bibliography{bib}

%
%
%
%
%
%

\end{document}